\documentclass[10pt, journal, compsoc]{IEEEtran}
\usepackage{amsmath}
\usepackage{bm, graphicx, bbm}
\usepackage{algorithm}
\usepackage{color}

\newcommand {\cA}{{\mathcal{A}}}

\newcommand {\cE}{{\mathcal{E}}}
\newcommand {\cG}{{\mathcal{G}}}

\newcommand {\cN}{{\mathcal{N}}}
\newcommand {\cO}{{\mathcal{O}}}

\newcommand {\cS}{{\mathcal{S}}}
\newcommand {\cT}{{\mathcal{T}}}
\newcommand {\cV}{{\mathcal{V}}}

\newcommand {\bM} {{\mathbbm M}}
\newcommand {\bN} {{\bf N}}
\newcommand {\bO} {{\bf O}}
\newcommand {\bX} {{\bf X}}
\newcommand {\bY} {{\bf Y}}

\newcommand {\bff} {{\bf f}}
\newcommand {\bd} {{\bf d}}
\newcommand {\be} {{\bf e}}

\newcommand {\bh} {{\bf h}}
\newcommand {\bn} {{\bf n}}
\newcommand {\bo} {{\bf o}}
\newcommand {\bp} {{\bf p}}

\newcommand {\bs} {{\bf s}}
\newcommand {\bt} {{\bf t}}
\newcommand {\bu} {{\bf u}}
\newcommand {\bx} {{\bf x}}
\newcommand {\by} {{\bf y}}

\newcommand {\bmu} {\boldsymbol{\mu}}

\newcommand{\avg}{{\rm avg}}

\newcommand {\Z} {{\mathbbm Z}}
\newcommand {\N} {{\rm I\kern-1.5pt N}}
\newcommand {\R} {{\rm I\kern-2.5pt R}}
\newcommand {\C} {{\rm I\kern-5pt C}}

\newtheorem{lemma}{Lemma}

\newtheorem{theorem}{Theorem}
\newtheorem{defn}{Definition}
\newtheorem{assm}{Assumption}

\newcommand{\beqa}{\begin{eqnarray}}
\newcommand{\eeqa}{\end{eqnarray}}
\newcommand{\beqan}{\begin{eqnarray*}}
\newcommand{\eeqan}{\end{eqnarray*}}
\newcommand{\beq}{\begin{equation}}
\newcommand{\eeq}{\end{equation}}

\newcommand{\bfl}{\begin{flushleft}}
\newcommand{\efl}{\end{flushleft}}
\newcommand{\bgnv}{\begin{verbatim}}
\newcommand{\ednv}{\end{verbatim}}

\newcommand{\myb}{\hspace{-0.1in}}
\newcommand{\mydef}{& \hspace{-0.1in} := & \hspace{-0.1in}}
\newcommand{\myeq}{& \hspace{-0.1in} = & \hspace{-0.1in}}

\newcommand{\lb}{\nonumber \\}
\newcommand{\myarr}{\begin{array}{lll}}
\newcommand{\indicate}[1]{{\bf{1}} \left\{#1\right\}}

\newcommand{\myleq}{& \myb \leq & \myb}

\newcommand{\bitem}{\begin{itemize}}
\newcommand{\eitem}{\end{itemize}}
\newcommand{\benum}{\begin{enumerate}}
\newcommand{\eenum}{\end{enumerate}}

\newcommand{\E}[1]{{\mathbbm E}\left[ #1 \right]}
\newcommand{\bP}[1]{{\mathbbm P}\left[ #1 \right]}
\newcommand{\myhb}{\hspace{-0.3in}}
\newcommand{\myhf}{\hspace{0.3in}}

\newcommand{\myskip}{\\ \vspace{-0.1in}}

\newcommand{\ER}{Erd$\ddot{\rm o}$s-R$\acute{\rm e}$nyi }

\newcommand{\change}[1]{{\color{black}{#1}}}

\def\QED{~\rule[-1pt]{5pt}{5pt}\par\medskip}
\newenvironment{proof}{{\bf Proof: \ }}{ \hfill \QED}

\ifCLASSOPTIONcompsoc
  \usepackage[nocompress]{cite}
\else
  \usepackage{cite}
\fi

\begin{document}
%
\title{Cascading Failures in Interdependent Systems: Impact of 
	Degree Variability and Dependence}


\author{Richard J. La\thanks{This work was 
supported in part by contracts
70NANB14H015 and 70NANB16H024 from National Institute
of Standards and Technology.} 
\thanks{Author is with the Department of Electrical \& 
Computer Engineering (ECE) and the Institute for Systems 
Research (ISR) at the University of Maryland, College Park.
E-mail: hyongla@umd.edu}
}


\IEEEtitleabstractindextext{%
\begin{abstract}
We study cascading failures in a system comprising 
interdependent networks/systems, 
in which nodes rely on other nodes both in the same 
system and in other systems to perform their function. 
The (inter-)dependence among
nodes is modeled using a dependence graph, where the 
degree vector of a node determines the number of other
nodes it can potentially cause to fail in {\em each} 
system through aforementioned dependency. 
In particular, 
we examine the impact of the variability and dependence 
properties of node degrees on the probability of
cascading failures. We show that larger 
variability in node degrees hampers widespread failures
in the system, starting with {\em random} failures. 
Similarly, positive correlations in node
degrees make it harder to set off an epidemic of 
failures, thereby rendering the system more robust
against random failures. 
\end{abstract}

}


\maketitle

\section{Introduction}
\label{sec:Introduction}

Many systems providing critical services to modern
societies (e.g., 
smart grids, manufacturing systems, transportation 
systems) comprise multiple 
heterogeneous systems that support each other
to enable the functionality of the overall system.
In particular, (local) decision makers or subsystems
belonging to different constituent or
component systems (CSes) rely on each other to 
perform various functions. 
For instance, a modern power system not only includes 
an electrical grid/network, but also depends on an 
information and communication network (ICN) to 
monitor the state of the electrical network and 
to communicate and execute appropriate control 
actions based on the observed state. 

Throughout the paper, we refer to the 
(local) decision makers or subsystems in various
CSes simply as {\em agents}.
Intricate (inter-)dependence among agents in CSes 
makes the analysis of these complex systems challenging. 
Moreover, in some cases, (random or targeted) 
failures of a small number of agents in one CS have
potential to cause unexpected, widespread failures 
of many agents in multiple CSes. 

The 2003 blackout in Italy provides
a good example~\cite{Rosato2008}.  
The onset of the blackout was 
triggered by an initial failure in the power grid, 
which caused a disruption to both the power grid 
and the ICN used for communication between power 
substations. This secondary failure in ICN further 
hampered the communication between stations and 
their regulation, setting off
rapid cascading failures throughout a large part
of the power grid.  

As illustrated by this example, due to increasing 
reliance of modern societies on such complex 
systems and interdependence among CSes, there is 
a growing interest in modeling and understanding 
the interaction between (agents in) interdependent 
CSes and the robustness of the overall systems
(e.g., \cite{Albert2000, Baxter2012, 
Buldyrev2010, Gao2011, Kenett2014, Rosato2008, 
Shao2011, Son2012, Vesp, Zhuang2016}). Yet, 
there is no theory that allows us to answer even 
a basic question, {\em ``Given two different 
interdependent networks or systems, when can
we say that one network or system is more
robust than the other?"}

The overarching goal of our study, which 
complements those of existing studies (summarized
in Section~\ref{sec:Related}), is to 
contribute to the emerging theory on complex 
systems, in particular on the influence of 
the {\em dependence structure properties} 
between agents
on the robustness of the systems with respect to 
{\em localized, random} failures in CSes. 
Our hope is that the findings will help engineers
and researchers identify critical properties
of robust systems and incorporate them into
{\em design guidelines} of complex systems. 

To this end, we develop a general model for capturing 
the propagation of failures from one agent to another 
both within individual CSes and across multiple CSes. 
This model is similar to that of \cite{Kham2016}
and allows us to introduce {\em asymmetric} dependence 
among agents belonging to {\em heterogeneous} systems 
(e.g., electrical network vs. 
ICN) and to study different ways in which failures 
can proliferate through diverse CSes.

Some key questions we are interested in are: (a) 
When is it possible for a localized initial failure
in one CS, beginning with one or a small number of 
randomly chosen agents, to trigger a cascade of failures 
not only within the CS in which the initial 
failure originated, but also in other CSes? 
(b) How does the underlying dependence 
structure among the agents in various CSes
influence the dynamics of failure propagation and 
the likelihood of such cascading failures? 
(c) How can we identify susceptible CSes that are more
likely to set off widespread failures across many 
CSes, starting with a few initial failures
in the CSes? In this study, we aim to offer 
partial answers to these important questions. 

While we carry out the study in the framework of
propagating failures in interdependent systems, 
we suspect that the basic model and approach as 
well as some
of key findings can be extended to other 
applications with appropriate changes. 
These applications include (i) information or rumor
propagation or new technology adoption via
multiple social networks, 
(ii) an epidemic of disease across multiples
geographic locations (e.g., cities or countries), 
and (iii) spread of malware in the Internet.

\subsection{Summary of main results}

We model the (inter-)dependence among agents 
with the help of what is known in the literature 
as a {\em degree-based} model or {\em Chung-Lu} 
model~\cite{ChungLu}. 
A similar model is used in many 
existing studies, e.g., 
\cite{Brummitt2012, Buldyrev2010, La_TON, 
Watts2002, Zhuang2016}. 
In order to capture different manners in which 
failures can spread both within various CSes
and between CSes, we model the dynamics of
failure propagation using a {\em 
multi-type branching process} with the
assumption that there are no cycles in the 
set of failed agents in a local neighborhood. 

We call other agents which an agent can 
influence and cause to fail 
its {\em (dependence) neighbors}, and 
model the number of neighbors that an agent
has in each CS, using a {\em random 
degree vector}. The $j$-th element in the 
degree vector of an agent belonging to  
the $i$-th CS ($i \neq j$) is the 
number of agents in the $j$-th CS which 
are dependent on the CS $i$ agent. 

Given fixed distributions of 
degree vectors for agents in different CSes, 
the tools from multi-type branching
process theory are employed to estimate 
the probability that a random initial failure 
in the $i$-th CS will give rise to cascading 
failures, affecting 
a large number of agents across the 
system. Since this probability of suffering
an epidemic of failures depends on which CS
suffers the initial failure, it also
tells us which CSes are more vulnerable to
random failures than others. 

The primary goal of our  study is to investigate 
how (i) the {\em variability of agents' degrees} 
(in the aforementioned random degree vectors) and 
(ii) the {\em dependence structure of the degrees} 
influence the likelihood of a random failure 
in a CS sparking a chain of failures throughout 
the system with {\em many} agents. 
To achieve this goal, we adopt well-known 
(integral) {\em stochastic  
orders} that are partial orders 
on the set of degree distributions
\cite{MullerStoyan, SO}. They 
allow us to compare degree distributions of
interest to us with respect to degree variability
and (positive) dependence.

There are many other important properties, such as 
assortativity \cite{Newman2002, Newman2003}
and clustering \cite{CoupLelarge, 
Zhuang2016} often observed in 
social networks or engineered systems as well 
as community structure \cite{Moharrami}, which 
influence the dynamics of information or failure 
propagation. But, as it will be clear, even without
these properties, analyzing the role of degree
variability and dependence is technically challenging.
For this reason, we do not model 
them here and leave an investigation of 
their effects for a future 
study. For instance, clustering is shown 
to impede global cascades in multiplex networks
\cite{Zhuang2016}, and thus the probability 
of cascades we estimate here may serve as
an upper bound when there is clustering. However, 
we suspect that the {\em qualitative} findings 
reported in this study will continue to hold even 
in the presence of clustering.

The high-level messages of our analytical
findings can be summarized as follows.
\myskip

\noindent {\bf F1.} 
Compare two distinct systems with 
different degree distributions of agents. Then, 
when there are a small number of random failures 
in some CS, the system in which agents' degrees 
exhibit {\em higher variability} is {\em less 
likely} to suffer widespread failures.
In particular, 
suppose that agents have identical or similar 
{\em average} degrees (hence comparable levels of 
dependence among agents and CSes)
in two different systems. In this case, 
the system with more homogeneous or 
predictable degrees (thus 
less variability in agents' degrees) is more 
susceptible to extensive failures in the
system.
\myskip

\noindent {\bf F2.} Consider two systems in 
which agents' degree distributions have {\em 
identical marginal distributions}. Therefore,
loosely speaking, we can say that they 
display the same level of (inter-)dependence. 
In this case, 
the system in which the agents' degrees are more
{\em positively correlated} is {\em less likely} 
to experience cascading failures as a result of 
initially localized, random failures. 
\myskip

The first finding indicates that degree 
distributions with higher variability, such as 
power laws, which permit the existence of large 
degree hubs, are more robust to random 
failures than more concentrated distributions 
such as Poisson distributions. 
Furthermore, it hints that systems in which 
all agents in each CS have similar degrees
are most prone to an outbreak of failures.
This observation is consistent with 
earlier {\em numerical} studies 
(e.g., \cite{Cohen1, Cohen2}) that suggest
that scale-free networks with power law
degree distributions are more resilient 
to {\em random} attacks, but vulnerable to
intentional attacks that target high-degree
nodes.  

The second finding above may be somewhat 
counter-intuitive at first sight. One might 
suspect that positive correlations 
would be helpful to spreading 
failures because high-degree agents are likely 
to have even larger aggregate 
degrees with increasing positive 
correlations and thus serve as more 
effective conduits for 
transmitting failures. However, our finding 
reveals that stronger positive correlations
have similar effects as higher variability
in the first finding. We suspect that the
reason behind this is that stronger 
positive correlations increase the 
variability in the aggregate degree of
agents. Consequently, they hinder the 
proliferation of failures, rendering the 
system more robust to random failures.

A few words on notation: throughout the paper, we will
use boldface letters or symbols to denote (row) 
vectors or vector functions.\footnote{All vectors 
are assumed to be row vectors.}  For instance, $\bd$ 
denotes a vector, and the $j$-th element of $\bd$ 
is denoted by $d_j$. Vectors ${\bf 0}$ and ${\bf 1}$ 
represent the vectors of zeros and ones, 
respectively, of appropriate dimensions. 
The set $\Z_+$ (resp. $\N$) denotes the set of 
nonnegative integers $\{0, 1, 2, \ldots\}$ (resp. 
positive integers $\{1, 2, 3, \ldots\}$).  Finally, 
all vector inequalities are assumed componentwise.

\section{Related Literature}
	\label{sec:Related}
	
There is already a large volume of 
literature on related topics, including 
cascading failures and robustness of 
complex systems
\cite{Albert2000, Boguna2003a, Boguna2003b, 
Vesp}, 
spread of epidemics and efficient immunization 
\cite{Pastor2005, Schneider2011, Yagan2012}, 
and information or rumor propagation
\cite{CoupLelarge, Zhuang2016}. 
Given the significant body of studies in related
fields, it is not possible to provide a summary 
of all. For this reason, we limit our discussion 
to a short list of most pertinent studies in 
the settings of multiplex or interdependent
networks, and do not discuss
other relevant studies 
(e.g., \cite{Baxter2012, Brummitt2012, 
Goltsev2008, Lee2014, Moharrami, Son2012}), 
including many important studies on a single, monolithic
network (e.g., \cite{Blume2011, Cohen1, Cohen2, 
CoupLelarge, La_TON, MotterLai, Motter, 
Newman2002, Newman2003, Watts2002}), 
here. We instead refer an interested reader to 
the references and those therein. 

In~\cite{Buldyrev2010}, Buldyrev et al. investigated
cascades of failures in two interdependent 
networks -- networks A and B -- using numerical 
studies. In their model, 
each node in network A (resp. B) depends on a
randomly chosen node in network B (resp. A),
which is modeled using a directed support link, 
and the failure of the node on which a node
depends causes its own failure. Furthermore, 
nodes in network A (resp. B) are connected with
each other according to a degree distribution 
$P_A$ (resp. $P_B$). Initially, a fraction 
($1 - p$) of network A nodes are removed, which 
triggers failures of nodes in both networks
through connectivity and 
dependence. They studied the probability
that a giant component survives as a function of
$p$, and identified a threshold for a first-order
phase transition. In addition, their numerical
results suggest that broader degree 
distributions $P_A$ and $P_B$ make the network
more vulnerable to {\em random} failures, whereas in a
single scale-free network, the opposite has
been observed. 

The findings in \cite{Buldyrev2010} have been 
extended in a series of follow-up studies. 
Parshani et al. \cite{Parshani2010}
studied a similar model and demonstrated that, as
the coupling between the two networks diminishes, 
the phase transition changes to a second-order
transition (from a first-order transition). 
In \cite{Shao2011}, Shao et al. 
relaxed the assumption that each node is 
dependent on exactly one node in the other
network and modeled the number of support links
of nodes using random variables. Huang et al.
\cite{Huang2011} considered the robustness of
the system against targeted attacks by mapping
the problem to a previously studied problem with
{\em random} attacks, and suggested
that the presence of high degree nodes in scale-free
networks makes it challenging to protect 
interdependent networks against targeted attacks. 


In~\cite{Zhuang2016}, 
Zhuang and Ya$\breve{\rm g}$an
studied information propagation in a multiplex
network with two layers representing 
an online social network (OSN) 
and a physical network, both
of which have high clustering. Only a subset of
vertices in the physical network are assumed to 
be active in the OSN.
Their key findings are: (a) clustering consistently 
hampers cascades of information to a large number of
nodes with respect to both the critical threshold
of information epidemics and the mean size of
epidemics; and (b) information transmissibility 
(i.e., average probability of information 
transmission over a link) has significant impact;
when the transmissibility is low, 
it is easier to trigger a cascade of information 
propagation with a smaller, densely connected 
OSN than with a large, loosely connected OSN. 
However, when the transmissibility is high, the 
opposite is true. 

In another study \cite{Hu2014}, Hu et al. studied
the problem of viral influence
spreading, for instance, in adoption of 
new technologies or scientific ideas. They 
modeled the spread of adoption using a 
multiplex network in which there are two
different types of links -- (i) undirected 
connectivity links and 
(ii) directed influence links. Outgoing
influence links of a node lead to other nodes whose
adoption of a new technology or idea causes
the node to adopt it with some fixed 
probability. Similarly, incoming influence
links of a node originate from other nodes that 
watch the node to see if it adopts the technology
or idea first and, if so, follow its trend 
with fixed probability. 

Their key findings include the following: 
(a) viral cascades are feasible only if there 
are positive correlations between the connectivity 
degrees and {\em outgoing} 
influence degrees of nodes. The intuition is that
when there are positive correlations, even the
adoption of a new technology by some random node
would make it easy to {\em influence} nodes
with high connectivity degrees because they
tend to have larger outgoing influence degrees, 
hence more likely to be influenced by other 
adopters; and (b) positive correlations between 
connectivity degrees and {\em incoming} influence 
degrees do not facilitate viral cascades much. 

Khamfroush et al.~\cite{Kham2016} investigated
the propagation of phenomena (e.g., failures, 
infections or rumors) in interdependent 
networks. By introducing temporal dynamics
into the model, they studied how quickly 
phenomena spread in three different types
of networks -- scale-free networks,
small-world networks, and  
\ER networks. Two of key observations
from their simulation studies
are: (a) scale-free networks are in general more
conducive to spreading phenomena than 
the other two types; and (b) the choices
of initial spreaders can have greater impact
than the network type. Based on the latter
observation, they proposed a new centrality
metric, called path-degree centrality, 
to better identify more effective initial
spreaders.

While many of these studies 
(e.g., \cite{Buldyrev2010, Huang2011, Parshani2010, 
Shao2011}) examine the robustness 
of multiplex or interdependent networks, 
their approaches and goals are very different
from those of our study. First, most of the
aforementioned studies focus on the analysis
of the emergence or survival of giant 
components under the assumption that only
the nodes that belong to the giant component
can continue to function properly. We, on the
other hand, investigate (i) when it is
possible to see an epidemic of failures through
dependencies among different systems 
and (ii) how the
{\em underlying dependence properties} between 
systems shape the likelihood of such catastrophic 
failures. 

Second, 
although the existing studies summarized here 
(and others we are aware of) provide interesting 
observations and major contributions to the growing 
understanding of complex systems, to the best of 
our knowledge, none of these studies aims to present 
{\em analytical} findings that enable us 
to {\em compare} the robustness of different 
networks on the basis of their dependence 
properties (which can be {\em partially ordered}). 
In contrast, as summarized in 
Section~\ref{sec:Introduction}, our intent is to
take another step towards building a {\em 
comprehensive theory} on complex systems 
which will help us determine when the robustness
of such systems improves or deteriorates as a
consequence of changes in their dependence 
properties. Some of our preliminary
results are reported in~\cite{La_CISS}.

\section{System Model}
\label{sec:Model}

Let $N_S > 1$ be the number of CSes in the (global) 
system we consider. For each $i \in \cS := \{1, 2, 
\ldots, N_S\}$, let $\cA_i$ be the set of agents in 
the $i$-th CS. When convenient, we use $a^i$ 
and $a^{i,k}$ to denote a generic agent in CS $i$ 
and the $k$-th agent in $\cA_i$, respectively.

We model the {\em internal} or {\em intra-CS} 
interdependence among 
agents in the $i$-th CS using a dependence graph 
$\cG_i = (\cV_i, 
\cE_i)$: the vertices in $\cV_i$ are the 
agents in CS $i$. 
The edges in $\cE_i$ are undirected edges between 
vertices in $\cV_i$ and indicate {\em mutual} 
dependence relations between the 
end vertices.\footnote{These {\em dependence 
relations} are not necessarily
the {\em physical links} in a network. For
example, in a power system, an overload failure
in one part of power grid can cause a failure
in another part that is not geographically close
or without direct physical connection to the former.} 
An undirected edge $e \in \cE_i$ should 
be interpreted as a pair of {\em directed} edges 
pointing in the opposite directions. Two agents
with an undirected edge between them are said to be
(dependence) neighbors. 

In addition to the (undirected) edges between agents
in the same CS, we model the dependence of an 
agent in one CS on another agent in a different 
CS using a {\em directed} edge;
if there is a directed edge from agent $a^{i}$ to
agent $a^{j}$, where $j \neq i$, 
this means that $a^j$ depends on $a^i$ and,
when $a^{i}$ fails, it could
cause $a^{j}$ to crash as well. 
We do not assume that this dependence is mutual to 
allow {\em asymmetric} dependence among 
CSes.\footnote{Mutual inter-CS dependence can be 
handled by replacing the directed edges with 
undirected edges similar to those used to model 
the intra-CS interdependence, and vice versa, 
without affecting our key findings. 
Also, a similar model is employed 
in \cite{Kham2016}.} 
If there is a directed edge $a^{i} \to a^{j}$
($j \neq i$), we say that $a^{j}$ is a CS 
$j$ (dependence) neighbor of $a^{i}$ 
and that $a^{i}$ {\em supports} $a^{j}$. 

We oftentimes need to distinguish the neighbors in 
the same CS from those in other CSes. For this reason,  
we call the neighbors 
in the same CS (resp. other CSes) {\em internal} 
neighbors (resp. {\em external} neighbors). Note
that an external neighbor of an agent is another 
agent in a different CS which it supports. 
In addition, we refer to the number of internal
neighbors of an agent as its {\em internal 
degree}.

\subsection{Agent degree distributions}

Throughout the paper, we shall use ${\bf D} 
= (D_j; \ j \in \cS)$ to denote a 
$N_S$-dimensional (random)
vector that describes the number of neighbors that 
an agent has in each CS $j \in \cS$. In other
words, $D_j$, $j \in \cS$, is the number of 
CS $j$ neighbors. We call ${\bf D}$ the {\em 
(dependence) degree vector} of the agent. Note 
that, for a CS $i$ agent,  
its degrees $D_j$, $j \neq i$, 
represent the number of agents that it
supports in other CSes and hence can affect 
in case of its own 
failure, but not those that support it.
Thus, they denote the {\em outgoing}
degrees of the agent.

The degrees of a CS $i$ agent are denoted by
a random vector ${\bf D}_i = (D_{i,j}; 
\ j \in \cS)$ with a 
distribution $p_i$; given 
${\bf d} \in \Z_+^{N_S}$, the probability 
that a randomly chosen CS $i$ agent
has $d_j$ neighbors in CS $j$, $j \in \cS$, 
is equal to $p_i({\bf d})$. We find it 
convenient to define the marginal distributions
$p_{i,j}: \Z_+ \to [0, 1]$, $j \in \cS$, where
$p_{i,j}(d_j) 
= \bP{\mbox{CS $i$ agent has $d_j$ CS 
	$j$ neighbors}}$.   

Note that the degrees of an agent 
to different CSes are {\em not} assumed to be
mutually independent. Put differently, 
the number of other agents that a CS $i$
agent, say $a^i$, supports in 
different CSes, i.e., ($D_{i,j}$; $j \in 
\cS \setminus \{i\}$), 
could be {\em correlated} and depend on its
internal degree $D_{i,i}$. This 
is important because in practice the failure of an 
important agent in a system may trigger the 
failure of many
other agents across different CSes, suggesting
that the degrees of such agents could be 
correlated. Thus, we wish to study the impact
of such degree correlations on system robustness. 

Finally, throughout this paper (except for in 
Section~\ref{sec:Two}), we assume that the
internal degree of an agent is at least one with
probability one, i.e., $p_{i,i}(0) = 0$ for
all $i \in \cS$; otherwise, the agent should
not belong to CS $i$.

\subsection{Propagation of failures}

To study the robustness of a system to
failures, we need to model how a failure spreads 
from one agent to another. Here, we explain the 
models we employ to approximate the
dynamics of failure propagation both within 
a CS and between agents
in different CSes. 
\myskip

\change{\noindent {\bf {\em P1. Intra-CS failure 
propagation --}} We model failure propagation 
within CS $i$ with the
help of a function $\wp^{in}_i: \N^2 \to [0, 1]$:
for fixed $d \in \N$ and $1 \leq n_f \leq d$, 
$\wp^{in}_i(d, n_f)$ tells us the 
probability that a CS $i$ agent with an internal 
degree $d$ will fail when $n_f$ internal 
neighbors collapse. 

An example that fits this model is the 
{\em random threshold} model used by Watts in
\cite{Watts2002}, which is also used
in other studies (e.g., \cite{Blume2011, 
Brummitt2012, Kham2016}). In the Watts' model, 
every agent $a^{i,k} \in \cA_i$ has 
some intrinsic value $\xi^{i,k}$. These values
$\xi^{i,k}$ of CS $i$ agents 
are modeled using mutually independent (continuous) 
random variables (rvs) 
with some common distribution $F_i$. We refer to
$\xi^{i,k}$ as their {\em security states}. 

In his model, a CS $i$ agent, say $a^{i}$, goes 
down as a consequence of the failures
of its internal neighbors when the fraction of its 
{\em failed} internal neighbors exceeds its 
security state $\xi^i$. Therefore, 
for a given pair $(d, n_f)$ with $n_f \leq d$,  
$\wp^{in}_i(d, n_f)$ is equal to $\bP{\xi^{i} < 
n_f / d} = F_i(n_f / d)$.
}
\myskip

\change{
\noindent {\bf {\em P2. Inter-CS failure 
propagation --}} We model the propagation of
a failure from one agent to an external 
neighbor in a similar manner. 
Suppose that agent $a^{j}$ is a CS $j$ agent.
Denote the number of CS $i'$ agents
that support $a^j$ by $D^{in}_{j,i'}$. We 
call $D^{in}_{j,i'}$ the {\em incoming} CS 
$i'$ degree of $a^j$ (to distinguish it from
its outgoing degree to CS $i'$).

Although our model can be generalized to 
allow the failure probability of $a^j$ (as a 
result of failures of CS $i$ agents supporting
$a^j$) to be a function of all of its incoming 
degrees ($D^{in}_{j,i'}$; $i' \in \cS \setminus 
\{j\}$) without altering our main results, for 
the ease of exposition, here we
adopt a simpler model in which the failure 
probability only depends on $D^{in}_{j,i}$:
when $a^j$ has incoming CS $i$
degree $D^{in}_{j,i}$ and $n^i_f$ of these
supporting CS $i$ agents collapse, the 
probability that $a^j$ also fails as a result 
is given by some function 
$\wp^{ex}_{j,i}(D^{in}_{j,i}, n^i_f)$.
Also, we assume that 
the external neighbors of a failed agent go down 
with the prescribed probability 
independently of each other.

It is clear that this model is general enough
to include the one studied
in \cite{Shao2011}, where $a^j$ is unaffected
by the failure of a supporting CS $i$ agent
as long as there is another supporting CS $i$ agent
that is still operational. 
}

\subsection{Tree-like infection graphs}

\change{
For our study, we focus on scenarios where 
collapsed agents do not cause the failures
of many neighbors on the 
average.\footnote{When each failed 
agent triggers many other neighbors to crash as well, 
cascading failures are likely and should happen often. 
This may indicate that
the system is poorly designed. Instead,
we are interested in more realistic scenarios of 
interest in which cascading failures are possible
and do occur, but not too frequently.} 
To make this more precise, we introduce {\em infection 
graphs}: starting with an initial failure 
of an agent in the system, the infection
graph consists 
of all {\em failed} agents and the {\em directed} 
edges used to contribute to the failures of 
neighbors. In other words, a directed edge 
from agent $a_1$ to agent $a_2$ belongs to the 
infection graph if and only if (iff) agent $a_2$ is 
a neighbor of $a_1$ and agent $a_1$
failed before agent $a_2$ did. We assume that this 
infection graph can be approximated 
using a tree-like structure, which we call 
an {\em infection tree}. A similar 
assumption is introduced in \cite{Brummitt2012, 
La_TON, Watts2002, Yagan2012}.
}

\change{
Although this assumption may not always hold in 
real systems, it allows us to approximate the 
dynamics of spreading failures as a multi-type 
branching process (described in Section
\ref{sec:MTBP}) and to carry out an 
analytical study. Moreover, when a graph is 
sparsely connected in large networks with small
average degrees, 
with high probability there are only few 
short cycles in the giant component
\cite{Janson}. Thus, it is a reasonable
assumption when failed agents do not
cause the crash of a large number of 
neighbors in each CS, which is the scenario
of interest to us.
}

\change{
When there is a directed edge from agent $a_1$
to agent $a_2$ in an infection tree, i.e.,  
the failure of $a_1$ causes that of $a_2$, 
we refer to $a_1$ (resp. $a_2$) as the {\em parent} 
(resp. a {\em child}). Also, borrowing from the 
language of epidemiology, we say that the parent 
{\em infected} the child.
}

\section{Agent Types and Children Distributions}

As mentioned in Section~\ref{sec:Introduction}, we
are interested in scenarios where the number of 
agents in each CS
is large. In a system with many agents,  
the propagation of failures can be approximated
with the help of a multi-type branching process
under some simplifying assumptions. 
To be more precise, we shall borrow from the theory
of branching processes with finitely many types
in order to study the likelihood of a small number
of initial failures leading to an epidemic
of failures infecting many other agents.




\subsection{Agent types}

In our model, depending on how an agent is infected, 
there are two possible types we need to consider 
for the failed agent. To formalize this, 
we define the {\em types} of agents in varying
CSes. Given $N_S$ CSes, 
there are $2 N_S$ types of interest to us. A 
CS $i$ agent ($i \in \cS$) can be either type 
$i$ or $i + N_S$: 
a type $i$ agent is a CS $i$ agent whose
internal neighbors are all 
functional, i.e., have not crashed. 
On the other hand, a CS $i$ agent is of type $i+N_S$ 
if it has an internal neighbor that 
went down. For notational simplicity, 
we use $i^+$ to denote $i + N_S$ ($i \in \cS$)
hereafter.

We shall discuss the distribution of the number of 
children of various types which are produced
by a failed agent of type $i \in \{1, 2, 
\ldots, 2 N_S\} =: \cS^+$ shortly. To explain
these children (vector) distributions, we first 
need to describe how we approximate the probability 
that a neighbor of a failed agent falls victim
to infection.

\subsection{Infection probability of neighbors}
	\label{subsec:InfProb}
	
\change{
The assumption that the infection graph is 
tree-like has an important implication: 
except for the root of the infection graph, 
each failed agent has exactly one
supporting agent whose collapse led to 
its own failure, namely its parent. 
This observation helps us simplify the models
used to capture the propagation of failures
as follows.

Because an agent faces possible infection from 
at most one failed supporting agent under 
the assumption, we can 
approximate the probability that a neighbor 
of a failed agent, say $a^i$, 
will be infected as explained below.

{\bf $\bullet$ Intra-CS infection probability --}
Consider an internal neighbor, $\tilde{a}^i$,
of agent $a^i$. 
Following the explanation in \cite{La_TON, 
Watts2002}, the probability that 
agent $\tilde{a}^i$ has $d$ internal neighbors 
is proportional to $d \cdot p_{i,i}(d)$. 
Recall that $\tilde{a}^i$ with internal degree $d$
will be infected by the failure of $a^i$ with 
probability $\wp^{in}_i(d, 1)$. Thus, by conditioning
on the internal degree of $\tilde{a}^i$, 
the probability that an 
internal neighbor of a failed CS $i$ agent
will be infected can be approximated using 
\beqa
\myb q_{i,i}
\mydef  \sum_{d \in \N} 
	\frac{d \cdot p_{i,i}(d)}{d_{i, \avg}} 
		\wp^{in}_i(d, 1)
	= \sum_{d \in \N} w_{i,i}(d) \ \wp^{in}_i(d, 1), 
	\label{eq:qii-1}
\eeqa
where $d_{i, \avg} := \sum_{d \in \N} d \cdot p_{i,i} (d)$ 
is the average internal degree of CS $i$ agents, 
and $w_{i,i}(d) := d \cdot p_{i,i}(d) / d_{i, \avg}$. 
In the example of Watts' model, we have 
$q_{i,i} = \sum_{d \in \N} w_{i,i}(d) \ F_i(d^{-1})$. 

The average internal degree $d_{i,\avg}$ in the 
denominator of (\ref{eq:qii-1}) serves as
a normalizing constant so that ${\bf w}_{i,i}
:= (w_{i,i}(d); \ d \in \N)$ is the internal degree
distribution of a randomly picked internal neighbor
\cite{La_TON, Watts2002}.\footnote{This 
sampling technique is called {\em sampling by random 
edge selection} \cite{LeskovecFaloutsos}.}
}

\change{
{\bf $\bullet$ Inter-CS infection probability --}
Suppose that a CS $j$ agent, $a^j$, is an external
neighbor of a failed CS $i$ agent $a^i$. Recall
that the inter-CS dependence is asymmetric.
Let $p^{in}_{j,i}$ be the conditional distribution 
of $D^{in}_{j,i}$ given that $a^j$ is a CS $j$
neighbor of a CS $i$ agent. 
Then, we approximate the probability that 
$a^j$ will be infected by the failure of $a^i$ 
using $q_{i,j} := \sum_{d \in \N} p^{in}_{j,i}(d) \ 
\wp^{ex}_{j,i}(d, 1)$.

Consider the model studied in \cite{Shao2011}
where the agent $a^j$ survives as long as one other 
supporting CS $i$ agent is functional. 
In this case, agent $a^j$ will be infected by 
the failure of $a^i$ iff $D^{in}_{j, i}
= 1$. In other words, $\wp^{ex}_{j,i}(1, 1) = 1$ and
$\wp^{ex}_{j,i}(d, 1) = 0$ for all $d > 1$. 
Thus, $q_{i,j} = p^{in}_{j,i}(1)$. 

The conditional distribution $p^{in}_{j,i}$
is determined by the distribution of {\em incoming} 
degrees $D^{in}_{j,i}$ of CS $j$ agents. In our study, 
we assume that the incoming degree distributions, 
hence their conditional distributions and 
$q_{i,j}$, $j \neq i$, are fixed while
we study the influence of the variability and dependence 
of internal and {\em outgoing} degrees of agents. 
}
\myskip

We believe that many of system parameters used in the
model can be estimated in practice, for instance, 
from historical data, physical 
laws (e.g., power grid), or simulation studies. These
include degree distributions $( p_i; \ i \in \cS)$ 
and infection probabilities $(q_{i,j}; \ i,j \in \cS)$, 
and can be used to examine the robustness of the system. 

\subsection{Distributions of children vectors}

Let $a^i$ be a randomly picked CS $i$ agent, 
and assume that it crashes
and infects some of its neighbors. 
Then, the type of its child
belongs to $(\cS \cup \{i^+\}) \setminus \{i\}
=: \cS^+_i$:
if $a^i$ triggers the failure of an external
neighbor in CS $j$, $j \neq i$, the type of the 
child is simply $j$. If $a^i$ causes an
internal neighbor to fail, then the child's 
type is $i^+$ because agent $a^i$ is an {\em 
infected} internal
neighbor of the child. 

Based on this observation, we can approximate 
the distribution of the number of children 
produced by $a^i$ by considering its two possible
types.

{\bf C1. Type $i$ agent $a^i$ ($i \in \cS$) --}
Let $D_{i,i}$ be its internal 
degree. Some of these internal neighbors, 
however, may not be affected by the failure
of agent $a^i$ and remain uninfected. 
Similarly, an external 
neighbor of agent $a^i$ in CS $j$ will 
go down (as a consequence of 
$a^i$'s failure) with probability $q_{i,j}$. 
As a result, the actual number of CS $j$ neighbors
infected by $a^i$, which we 
denote by $O_{i,j}$, can differ from $D_{i,j}$, 
$j \in \cS \setminus \{i\}$.



Summarizing this argument, 
we approximate 
%
the probability that $a^i$ 
will produce children given by a children
vector ${\bf o} = (o_j; \ j \in \cS^+) \in 
\Z_+^{2 N_S}$, where $o_j$ is the number of 
type $j$ children, using the following 
children distribution:
\beqa
&& \myhb h_i({\bf o}) 
	\label{eq:hi1} \\
\myeq \left\{ \begin{array}{l}
	\sum_{{\bf d} \geq \bar{\bd}({\bf o})} 
	\Big( p_{i}(\bd) \\
	\myhf \times \prod_{j \in \cS} 
		{{d_j}\choose{\bar{d}_j(\bo)}} 
			\left( q_{i,j} \right)^{\bar{d}_j(\bo)}
			\left( \bar{q}_{i,j} \right)^{d_j 
				- \bar{d}_j(\bo)} \Big) \lb
	\hspace{0.15in} \mbox{ if } o_j = 0
	\mbox{ for all } j \notin \cS^+_i, \\
	0 \hspace{0.12in} \mbox{otherwise,}
	\end{array} \right.
	\nonumber
\eeqa
where $q_{i,j}$ are as defined earlier,
$\bar{q}_{i,j} := 1 - q_{i,j}$ for all $j 
\in \cS$, and $\bar{\bd} : \Z_+^{2 N_S} \to 
\Z_+^{N_S}$ with 
\beqa
\bar{d}_j({\bf o})
\myeq \left\{ \begin{array}{ll}
	o_j & \mbox{if } j \neq i, \\
	o_{i^+} & \mbox{if } j = i.
	\end{array} \right.  
	\label{eq:bard}
\eeqa

{\bf C2. Type $i^+$ agent $a^i$ --}
For type $i^+$ agents ($i \in \cS$), the 
children distribution is closely related to that of 
type $i$ agents with a minor difference: for ${\bf o}
\in \Z_+^{2 N_S}$,
\beqa
&& \myhb h_{i^+}({\bf o}) 
	\label{eq:hi2} \\
\myeq \left\{ \begin{array}{l}
	\sum_{\bd \geq \bar{\bd}({\bf o})} 
	\Big( p_{i}(\bd + \be_i) \\
	\myhf \times \prod_{j \in \cS} 
		{{d_j}\choose{\bar{d}_j(\bo)}} 
			\left( q_{i,j} \right)^{\bar{d}_j(\bo)}
			\left( \bar{q}_{i,j} \right)^{d_j 
				- \bar{d}_j(\bo)} \Big) \lb
	\hspace{0.15in} \mbox{ if } o_j = 0
	\mbox{ for all } j \notin \cS^+_i, \\
	0 \hspace{0.12in} \mbox{otherwise,}
	\end{array} \right.
	\nonumber
\eeqa
where ${\bf e}_i$ is a zero-one vector whose only nonzero
entry is the $i$-th entry.

The only difference between (\ref{eq:hi1})
and (\ref{eq:hi2}) is that, for type $i^+$ agents, 
we first remove one of
internal neighbors before counting the
neighbors that can be infected by the agents. 
The reason for this is that a type $i^+$ 
agent has a parent in CS $i$ and the number of 
{\em remaining} internal neighbors that it can 
potentially infect is its internal degree minus one.
We denote the marginal distribution of the number 
of type $j$ children of a type $i$ agent by $h_{i,j}$, 
$i, j \in \cS^+$. 
\myskip

Before we proceed, let us comment on a simplifying 
assumption we implicitly introduced in 
Section~\ref{subsec:InfProb} and our 
approximations in (\ref{eq:hi1}) and (\ref{eq:hi2}).  
Suppose that $a^i$ is a CS $i$ agent that is 
infected by an internal neighbor. Given
that $a^i$'s failure is caused by another
internal neighbor, its conditional degree distribution 
will likely be different from $p_i$ we assumed in
(\ref{eq:hi1}) and (\ref{eq:hi2}). 

For instance, if $d \cdot \wp^{in}_i(d, 1)$ is increasing
(resp. decreasing), its internal degree $D_{i,i}$
{\em conditional} on the event that it is 
infected by an internal neighbor, is larger 
(resp. smaller) than the unconditional internal 
degree with respect to the usual 
stochastic order~\cite{SO}. The reason for this is
that the internal degree distribution of an internal
neighbor of $a^i$ is given by ${\bf w}_{i,i}$, 
where $w_{i,i}(d) \propto d \cdot p_{i,i}(d)$.
Therefore, if $d \cdot \wp^{in}_i(d, 1)$ is increasing
in $d$, we have 
\beqan
\frac{w_{i,i}(d) \wp^{in}_i(d, 1)}{p_{i,i}(d)}
\leq \frac{w_{i,i}(d+1) \wp^{in}_i(d+1, 1)}{p_{i,i}(d+1)}
\ \mbox{ for all } d \in \N.
\eeqan 

Moreover, this conditional degree distribution could 
differ from the conditional degree distribution
we would see provided that $a^i$ is infected by a 
parent in a different CS, which would also depend
on the CS to which the parent belongs. 
Therefore, it is clear that 
the conditional degree distribution of an infected 
agent will likely differ from $p_i$ and depend on 
how it was infected. 

Unfortunately, computing and adopting accurate 
conditional degree distributions for the 
analysis is quite challenging for several 
reasons. For example, in order to compute the 
necessary conditional probabilities, 
for each failed agent $a^i$, 
we need to know exactly how it was 
infected. More precisely, 
we must take into account the {\em history} 
or the sequence
of agents that collapsed and led to the
infection of $a^i$, as well as the 
joint distributions of agents' incoming
and outgoing degrees.
The reason for this is that the 
conditional degree distribution 
of the {\em parent} of $a^i$ in turn depends on 
its parent and so on. 

Iteratively computing
the conditional degree distributions of
{\em all} infected agents while accounting for 
the history and relevant joint degree
distributions quickly becomes intractable.
For this reason, in order to maintain
mathematical tractability of the model, we make 
a simplifying assumption that the 
(conditional) degree
distribution of {\em infected} CS $i$
agents can be approximated by $p_i$ for
all $i \in \cS$. However, we believe that
this is a reasonable assumption, especially
when we compare two systems with similar 
degree distributions for {\em local comparison} 
(as the assumption or its failure would affect 
both of them alike).

\section{Multi-type Branching Process for Modeling
	the Spread of Failures} 
	\label{sec:MTBP}
	
We approximate the propagation of failures, using 
a multi-type branching process: when a type $i$
agent ($i \in \cS^+$) fails, it produces 
children of various types in accordance with the
children distribution described in 
the previous section, {\em independently of other 
infected agents in the system}. 

\subsection{Infection tree}

Suppose that a CS $i$ agent ($i \in \cS$), 
say $a^i$, is the first agent to experience
a random failure. 
As mentioned in Section~\ref{sec:Introduction},
we are
interested in determining: (i) if it is possible 
for this single failure to lead to widespread 
infection of many other agents through 
dependence among agents in multiple
CSes, and (ii) if so, how likely it is for the
system to suffer such a cascade of failures.

To answer these questions, we consider a
(directed) infection tree that captures 
the spread of failures, which is rooted at 
agent $a^i$. We denote the tree by $\cT
:= (\cV_{\cT}, \cE_{\cT})$, where $\cV_\cT$ is 
the set of failed agents and $\cE_\cT$ is
the set of directed edges via which failures 
transmitted. 
For each $k \in \Z_+$, let $\cN(k)
= (\cN_j(k) ; \ j \in \cS^+)$ 
denote the set of $k$-hop neighbors of 
$a^i$ in $\cT$, where 
$\cN_j(k)$ is the set of type $j$ 
$k$-hop neighbors. 

\begin{figure}[h]
\centering
\includegraphics[width=0.4\textwidth]{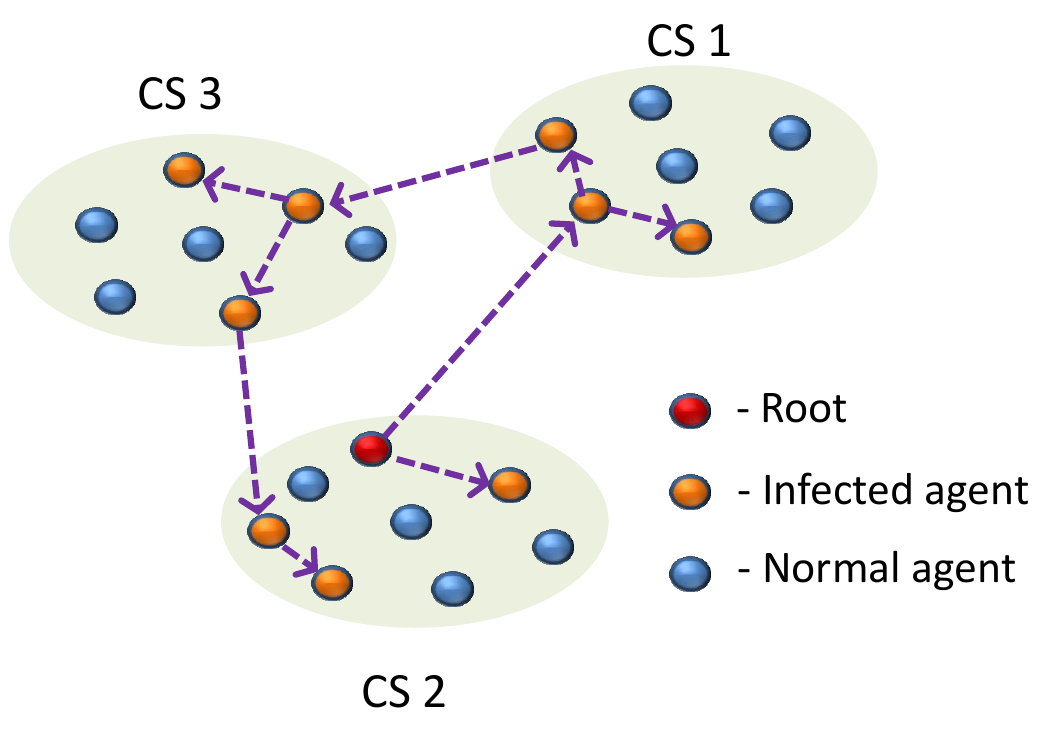}
\caption{Propagation of failures.}
\label{fig:tree}
\end{figure}

Note that $\cN(k)$, $k \in \N$, are random sets, 
and we are mostly interested in the cardinalities 
$N_j(k) = | \cN_j(k) |$.
In the example with three CSes ($N_S = 3$)
shown in Fig.~\ref{fig:tree}, 
the initial failure occurs in CS 2 (Root 
shown as a filled red circle). The dotted
arrows indicate how the failures transmitted between
agents. Here, $N_j(1) = 1$ for $j \in \{1,5\}$
and  $N_j(1) = 0$ for $j \in \{2, 3, 4, 6\}$. 
Similarly, $N_4(2) = 2$ and $N_j(2) = 0$
for $j \in \{1, 2, 3, 5, 6\}$. The tree ${\cal T}$
consists of the root and infected agents
(filled orange circles) along with the
dotted arrows. 
 
Regrettably, computing the exact distribution 
of the total number of failed agents (i.e., 
$1 + \sum_{k \in \N} \left( \sum_{j \in \cS^+}
N_j(k) \right)$) is 
challenging, if possible at all, for large systems. 
For this reason, we follow a similar approach
employed in \cite{Brummitt2012, La_TON, Watts2002,
Yagan2012} and, rather than analyzing a  
finite system, consider an infinite system 
in which the degree vector of {\em each} CS $i$ 
agent is given by a random vector with a common
distribution $p_i$, {\em independently of each other}. 
In other words, the degree
vectors of CS $i$ agents are given 
by independent and identically distributed (i.i.d.)
random vectors with the distribution $p_i$. 
Moreover, the degree vectors of agents in different
CSes are assumed mutually independent. This 
degree-based model is also known as the
Chung-Lu model in the literature~\cite{ChungLu}. 

By the strong law of large numbers, the fraction 
of CS $i$ agents with degree vector $\bd$ converges 
to $p_i(\bd)$ almost surely for all $\bd \in 
\Z_+^{N_S}$. Using this model, we will first look 
for a condition under which $\bP{
\limsup_{k \to \infty} N_j(k) = 0 \mbox{ for
all } j \in \cS^+} < 1$.
Put differently, there is {\em positive}
probability that the failures will continue to
propagate forever 
in an infinite system. We shall use 
this probability of cascading failures 
(PoCF) in an
infinite system to approximate the probability
that a large system would experience an epidemic
of failures. 

The answer to this question can be obtained by
studying a multi-type branching process with 
$2 N_S$ types. Let ${\bf N} = \{{\bf N}(k), \ 
k \in \Z_+\}$, where ${\bf N}(k) = (N_i(k); \ 
i \in \cS^+)$
and $N_i(k)$, $i \in \cS^+$, is the number of 
type $i$ agents in the $k$-th generation. 
Recall that, for $i \in \cS$, 
a type $i$ agent (resp. $i^+$ agent) 
is a CS $i$ agent with no infected internal 
neighbor (resp. with a failed internal
neighbor). 


\subsection{Probability of extinction}

The probability $\bP{ \limsup_{k \to \infty} 
N_j(k) = 0 \mbox{ for all } j \in \cS^+}$ is 
called the probability of extinction (PoE)
\cite{Harris}. Obviously, the PoCF 
is equal to one minus the 
PoE. Since the initial failure
can originate in any of $N_S$ CSes, we denote 
the PoE starting with a random failure in CS 
$i$ by $\mu_i$.

For each $i \in \cS^+$, let
$\E{h_i} = (\E{h_{i,j}}; \ j \in \cS^+)$
be a $1 \times 2 N_S$ row 
vector,
whose $j$-th element is the expected number of
type $j$ children from a failed type 
$i$ agent. Define $\bM = [M_{i,j}]$ to be 
a $(2 N_S) \times (2 N_S)$ matrix, whose
$i$-th row is $\E{h_i}$, i.e., 
$M_{i,j} = \E{h_{i,j}}$ for all $i, j
\in \cS^+$. 

An example of $\bM$ with $N_S = 2$ is shown below. 
\beqan
\bM
\myeq \left[ 
	\begin{array}{cccc}
	0 & M_{1,2} & M_{1,3} & 0 \\
	M_{2,1} & 0 & 0 & M_{2,4} \\
	0 & M_{3,2} & M_{3,3} & 0 \\
	M_{4,1} & 0 & 0 & M_{4,4} \\
	\end{array}
\right]
\eeqan
Note that $M_{i,i} = M_{i^+, i} = 0$ for all 
$i \in \cS$ because a CS $i$ agent infected 
by another CS $i$ agent will be type 
$i^+$ ($i \in \cS$). 
Similarly, $M_{i, j^+} = 0$ for all $i, j \in 
\cS$ and $j \neq i$ under the assumption
of no cycles in the infection tree.  

\begin{defn}	\label{def:pos-regular}
A square matrix ${\bf A}$ is said to be 
{\em (positively) regular} if there exists
$k \in \N$ such that ${\bf A}^k$ is
positive, i.e., all entries are positive. 
\end{defn}

\begin{assm}	\label{assm:pos-regular}
We assume that $\bM$ is (positively) regular.
\myskip
\end{assm}

One can show that a {\em sufficient} condition 
for the (positive)
regularity of $\bM$ is that (i) $p_{i,i}(0) + 
p_{i,i}(1) < 1$
for all $i \in \cS$ and (ii) $\bM$ is 
irreducible.\footnote{Let ${\bf A} = 
[A_{i,j}]$ be a $2 N_S \times 2 N_S$ matrix, 
where $A_{i,j} = \indicate{M_{i,j} > 0}$ for 
all $i, j \in \cS^+$. This matrix ${\bf A}$ 
can be viewed as an adjacency matrix for a 
directed graph ${\cal G}$ with $2 N_S$ 
vertices that represent the $2 N_S$ types. 
There is a directed edge from vertex $i$
to vertex $j$ if a type $i$ agent produces
a type $j$ agent with positive probability. 
The matrix $\bM$ is irreducible if
and only if the directed graph 
${\cal G}$ is {\em strongly connected}
\cite{AGT}.} 
The first condition simply
means that when a CS agent $i$ fails, 
there is positive probability that it will
infect another internal neighbor. 
The second condition ensures that a random 
initial failure
in {\em any} CS can eventually cause
some agents in every other CS $j$, $j \neq
i$, to go down with positive probability,
following a sequence of infections. 

If $\bM$ is regular, there exists
$k^\star \in \N$ such that, starting with any 
random failure, regardless of which CS
experiences the initial failure, there is
strictly positive probability that 
$\bN(k^\star) \geq {\bf 1}$, i.e., there is
a failure in {\em every} CS. 

Let $\boldsymbol{\mu} = (\mu_i; \ i \in 
\cS^+)$. Although 
the PoEs of interest to us are ($\mu_i$; 
$i \in \cS$), we will compute ($\mu_{i^+}$;
$i \in \cS$) as well.
Under the (positive) regularity assumption, 
$\boldsymbol{\mu} = {\bf 1}$ if (i) $\rho(\bM)
< 1$ or (ii) $\rho(\bM) = 1$ and there is
at least one type for which the probability 
that it produces 
exactly one child is not equal to one, 
where $\rho(\bM)$ is the spectral
radius of $\bM$~\cite{matrix}. Similarly, 
if $\rho(\bM) > 1$, then $\boldsymbol{\mu}
< {\bf 1}$ and there is strictly positive
probability that the cascading failures  
continue forever in an infinite system, suggesting
that there could be an epidemic of failures
in a large system.

It is noteworthy that whether or not there
could be cascading failures in the infinite
system depends only on the 
mean number of children of varying types
that each type produces.
However, the exact
PoEs $\bmu$ vary from one set of children
distributions $\{h_i; i \in \cS^+\}$ to another 
set $\{\tilde{h}_i; i \in \cS^+\}$
even when the matrix $\bM$ remains the same.

For each $i \in \cS^+$, define a generating
function $f_i: [0, 1]^{2 N_S} \to [0, 1]$, 
where 
\beqan
f_i(\bs) = \sum_{\bo \in \Z_+^{2 N_S}} 
	h_i(\bo) \prod_{j \in \cS^+} s_j^{o_j}.
	\label{eq:fi}
\eeqan
Then, the PoE vector $\boldsymbol{\mu}$ is given as a 
fixed point that satisfies 
\beqa
{\bf f}(\bmu)
\myeq \bmu \leq {\bf 1}, 
	\label{eq:fixed}
\eeqa
where ${\bf f}(\bs) = (f_1(\bs), \ldots, f_{2N_S}(\bs))$.
When $\rho(\bM) > 1$, there exists a unique $\bmu$
that satisfies (\ref{eq:fixed}) with strict
inequality, i.e., $\bmu < {\bf 1}$
\cite{Harris}. 

A key question of interest to us is how the degree 
distributions ($p_i$; $i \in \cS$) affect the PoE 
vector $\boldsymbol{\mu}$, especially with 
{\em fixed} $\bM$. To be more precise, we
will investigate how the {\em variability} and 
{\em dependence structure} of agents' degree vectors
shape the PoEs.  
To this end, we introduce
several stochastic and dependence orders that we
employ to compare the degree distributions. 
Using these orders, we first examine a simple
scenario consisting of two {\em symmetric} 
interdependent CSes in the subsequent section, 
followed by more general 
settings in Section~\ref{sec:General}.

\section{Two symmetric interdependent systems}
	\label{sec:Two}

The goal of this section is, by studying 
simpler scenarios first, to highlight 
some insights on how (i) the variability of 
degrees of agents
(i.e., the number of neighbors in two different
CSes) and (ii) the dependence of the two degrees
influence the PoEs, even when the mean degrees
remain fixed. 

Consider a system comprising two interdependent 
CSes ($N_S = 2$), 
and suppose that the degree distributions 
$p_1$ and $p_2$ are symmetric, i.e., 
$p_1(d_1, d_2) = p_2(d_2, d_1)$ for 
all $(d_1, d_2) \in \Z_+^2$. Moreover, in order to 
simplify the analysis and shed some light on 
our main findings in general settings to follow, 
we set $q_{i,j} = 1$ for all 
$i, j \in \cS$. This assumption will be relaxed 
in the subsequent section. To 
correctly interpret this assumption 
and the findings in this section, 
a reader should view the symmetric degree 
distributions $p_i$, $i = 1, 2,$ as children 
distributions $h_i$, $i = 1, 2,$ instead; 
otherwise, if all internal dependence graphs 
${\cal G}_i$, $i \in \cS$, are connected and 
$q_{i,j} = 1$ for all $i,j \in \cS$, 
under positive regularity assumption of $\bM$, 
every agent will eventually be infected, 
starting with any failure. For this reason, 
we remove the assumption that the internal
degree is at least one in this section
to allow for the possibility that some agents
do not produce any children. This assumption
will be reintroduced in the following section.

Throughout this and following sections, 
we assume that $\rho(\bM) > 1$
under all considered degree distributions
so that it is possible for a random failure 
to trigger cascading failures. 
In the case of two interdependent CSes, 
the assumed symmetry of the degree 
distributions and the uniqueness of the 
fixed point $\bmu$ satisfying $\bff(\bmu) = 
\bmu < {\bf 1}$
tell us $\mu_1 = \mu_2$ and $\mu_3 = \mu_4$.

\subsection{The effects of degree variability}

We first study the variability of degrees. One 
common way to compare the variability of two 
rvs is the {\em second-order stochastic
dominance} (SSD)~\cite{SO}. Loosely speaking, 
if rv $Y$ dominates rv $X$ with respect to
SSD ($X \leq_{SSD} Y$) and $\E{X} = \E{Y}$, 
it means that $Y$ is more {\em predictable}
than $X$. It turns out $X
\leq_{SSD} Y$ is equivalent to 
$X$ being smaller than $Y$ with respect to 
{\em increasing concave} (ICV) order 
($X \leq_{icv} Y$);
for all increasing, concave
functions $\xi: \R \to \R$, $\E{\xi(X)}
\leq \E{\xi(Y)}$.\footnote{These inequalities
in the definitions of various stochastic
and dependence orders are required to hold 
only for the functions
for which the expectations are well defined.} 
Since $\xi(x) = x$ is concave and
increasing, $X \leq_{SSD} Y$ implies $\E{X} 
\leq \E{Y}$.

In order to eliminate the effects of the 
correlations between two degrees 
and focus on the role of their 
variability on PoEs, we assume that the two
degrees of an agent are independent in this
subsection. 

\begin{assm}	\label{assm:independence}
$p_{1}(d_1, d_2) = p_{1,1}(d_1) \cdot
p_{1,2}(d_2)$ for all $d_1, d_2 \in \Z_+$. 
\myskip
\end{assm}

The following lemma illustrates how the {\em variability} 
in degrees affects the PoEs when the degrees of an 
agent are independent.

\begin{lemma}	\label{lemma:L1}
Consider two degree distributions $p^{(\ell)}_1$, 
$\ell = 1, 2$.
Let ${\bf D}^{(\ell)}_1$, 
$\ell = 1, 2,$ be a random vector with 
distribution $p^{(\ell)}_1$. 
Suppose that Assumption
\ref{assm:independence} holds for $p_1^{(\ell)}$, 
$\ell = 1, 2$, and $D^{(1)}_{1, j} \leq_{SSD}
D^{(2)}_{1, j}$ for $j = 1, 2$. 
Then, $\bmu^{2} \leq \bmu^{1}$, where 
$\bmu^{\ell} < {\bf 1}$, $\ell = 1, 2$, is the PoE
vector under degree distribution $p^{(\ell)}_1$.
\end{lemma}
\begin{proof}
A proof is provided in Section~\ref{appen:L1}.
\end{proof}

The lemma tells us that $\mu^{2}_1 = 
\mu^{2}_2 \leq \mu^{1}_1 = \mu^{1}_2$. Thus, 
an implication of Lemma~\ref{lemma:L1} is that
even when the mean degrees of agents are fixed, 
the PoEs tend to increase as the degrees of 
agents become more spread out, i.e., have greater 
variability, suggesting that widespread failures
would be {\em less likely} as the degrees of 
agents vary more widely. 

We say that $X$ is smaller than $Y$ with 
respect to first-order stochastic dominance
(FSD) or usual stochastic order if, for all
increasing functions $\xi: \R \to \R$, 
$\E{\xi(X)} \leq \E{\xi(Y)}$ 
\cite{MullerStoyan, SO}. 
This is equivalent to $F_X(t) \geq F_Y(t)$
for all $t \in \R$. 
Clearly, by definition, FSD implies SSD. 
Hence, $D^{(1)}_{1,j} \leq_{FSD}
D^{(2)}_{1,j}$, $j = 1, 2$, is a sufficient
condition for Lemma~\ref{lemma:L1} to hold 
and, as one would expect, when agents' degrees 
become larger, an outbreak of failures is
more likely. 

\begin{table}[h]
\begin{center}
\begin{tabular}{r c}
& $d_1$ \\
$d_2$ & 
\begin{tabular}{|r| c c c c|}
\hline
& 0 & 1 & 2 & 3 \\ \hline
0 & 0.40000 & 0.12000 & 0.16000 & 0.01200 \\
1 & 0.03750 & 0.01125 & 0.01500 & 0.01125 \\
2 & 0.05000 & 0.01500 & 0.02000 & 0.01500 \\
3 & 0.01250 & 0.00375 & 0.00500 & 0.00375 \\
\hline
\end{tabular}
\end{tabular} \\
Degree distribution $p^{(1)}_1(d_1, d_2)$ \\
\vspace{0.1in}

\begin{tabular}{r c}
& $d_1$ \\
$d_2$ & 
\begin{tabular}{|r| c c c c|}
\hline
& 0 & 1 & 2 & 3 \\ \hline
0 & 0.40000 & 0.04000 & 0.32000 & 0.04000 \\
1 & 0.02500 & 0.00250 & 0.02000 & 0.00250 \\
2 & 0.07500 & 0.00750 & 0.06000 & 0.00750 \\
\hline
\end{tabular}
\end{tabular} \\
Degree distribution $p^{(2)}_1(d_1, d_2)$
\vspace{0.1in}

\begin{tabular}{r c}
& $d_1$ \\
$d_2$ & 
\begin{tabular}{|r| c c c c|}
\hline
& 0 & 1 & 2 & 3 \\ \hline
0 & 0.40000 & 0.04700 & 0.32000 & 0.03300 \\
1 & 0.02500 & 0.00250 & 0.02000 & 0.00250 \\
2 & 0.07500 & 0.00050 & 0.06000 & 0.01450 \\
\hline
\end{tabular}
\end{tabular} \\
Degree distribution $p^{(3)}_1(d_1, d_2)$
\end{center}
\caption{Degree distributions $p^{(\ell)}_1$, 
$\ell = 1, 2, 3$.}
\vspace{-0.1in}
\label{table:1}
\end{table}

{\bf Example 1:} 
Consider two degree distributions 
$p^{(1)}_1$ and $p^{(2)}_1$ shown in Table
\ref{table:1}. Even though these distributions 
may not be realistic or representative, we use them 
to illustrate our findings with numerical
examples. One can 
easily verify that (i) $D^{(\ell)}_{1,1}$ and 
$D^{(\ell)}_{1,2}$, $\ell = 1, 2$, are independent 
and (ii) $D^{(1)}_{1,j} \leq_{SSD}
D^{(2)}_{1,j}$, $j = 1, 2$. In addition, both
distributions yield 
\beqan
\bM
\myeq \left[ 
	\begin{array}{cccc}
	0 & 0.35 & 1.00 & 0 \\
	0.35 & 0 & 0 & 1.00 \\
	0 & 0.35 & 0.50 & 0 \\
	0.35 & 0 & 0 & 0.50 
	\end{array}
\right]
\eeqan
with $\rho(\bM) = 1.021 > 1$. The 
entropies of $D^{(1)}_1$ and $D^{(1)}_2$
(resp. $D^{(2)}_1$ and $D^{(2)}_2$) are
1.786 and 1.003 (resp. 1.461 and 0.884), 
respectively, suggesting that $D^{(1)}_j$, 
$j = 1, 2$, are more unpredictable than 
$D^{(2)}_j$, $j = 1, 2$.  

The PoE vector satisfying (\ref{eq:fixed}) for 
$p^{(1)}_1$ (resp. $p^{(2)}_1$) is $\bmu^1 = 
(0.9646 \ 0.9646 \ 0.9761 \ 0.9761)$
(resp. $\bmu^2 = (0.9586 \ 
0.9586 \ 0.9720 \ 0.9720)$).
Thus, although the two degree distributions yield
the same matrix $\bM$, the PoE is larger
under distribution $p^{(1)}_1$. Equivalently, 
the PoCF, beginning
with a random failure in either CS, 
is 0.0354 (resp. 0.0414) 
under $p^{(1)}_1$ (resp. $p^{(2)}_1$), which
represents roughly a 17 percent difference
in PoCF.

\subsection{The effects of degree dependence}

We now turn our attention to the role of 
dependence between the two degrees of an agent. 
To this end, we adopt a well-known dependence
order, called {\em concordance order} (CO)
\cite{MullerStoyan}: suppose that $\bX
= (X_1, X_2)$ and $\bY = (Y_1, Y_2)$ are
two bivariate rvs {\em with identical marginal 
distributions}. This 
means that the variability of each rv 
remains fixed. Then, $\bX$ is smaller than 
$\bY$ in CO $(\bX \leq_C \bY)$ if 
Cov($\xi_1(X_1), \xi_2(X_2)$) $\leq$
Cov($\xi_1(Y_1), \xi_2(Y_2)$) for all 
increasing functions $\xi_j$, $j = 1, 2$. 
Note that this implies Cov($X_1$, $X_2$)
$\leq$ Cov($Y_1, Y_2$). 

Roughly speaking, 
$\bX \leq_C \bY$ means that $Y_i$, $i = 1, 
2$, are more positively correlated than $X_i$, 
$i = 1, 2$. In addition, as explained 
in \cite[p. 109]{MullerStoyan}, CO
is the only integral stochastic order that
satisfies natural properties that one would
expect of a stochastic order for comparing
dependence.

The second lemma examines how the (positive)
dependence of degrees influences the PoEs. 
Its proof is omitted here due to a 
space constraint and can be found in 
\cite{La_arvix}.  

\begin{lemma} \label{lemma:L2}
Consider two degree distributions $p^{(\ell)}_1$, 
$\ell = 1, 2$,
with identical marginal distributions. 
Let ${\bf D}^{(\ell)}_1$, 
$\ell = 1, 2,$ be a random vector with
distribution $p^{(\ell)}_1$. Suppose 
${\bf D}^{(1)}_1 \leq_{C} {\bf D}^{(2)}_1$. 
Then, $\bmu^{1} \leq \bmu^{2}$, where 
$\bmu^{\ell}$, 
$\ell = 1, 2$, is the PoE vector under the 
degree distribution $p^{(\ell)}_1$.
\end{lemma}

A key finding of Lemma~\ref{lemma:L2} is
that as the two degrees of agents become
more positively correlated, it becomes 
more difficult to set off cascading failures. 
One possible way to 
interpret this finding is that as the 
degrees become more positively correlated, 
the variability in the total degree of an agent, 
i.e., the sum of two degrees, also grows. 
Hence, Lemma~\ref{lemma:L1} suggests that 
the PoEs should increase.  
\myskip

{\bf Example 2: } For the second 
example, consider degree distributions 
$p^{(2)}_1$ and $p^{(3)}_1$ given in 
Table~\ref{table:1}. To obtain $p^{(3)}_1$, we 
modified $p^{(2)}_1$ in order to introduce
weak positive correlations between the two 
degrees by (i) adding 0.007 to $p^{(2)}_1(1, 0)$
and $p^{(2)}_1(3, 2)$ and (ii) subtracting 
0.007 from $p^{(2)}_1(3, 0)$ and
$p^{(2)}_1(1, 2)$. The correlation coefficient
of $D^{(3)}_{1,1}$ and $D^{(3)}_{1,2}$ is 0.0368, 
indicating weak positive correlations.

One can show that (i) $p^{(2)}_1$ and 
$p^{(3)}_1$ have the same 
marginal distributions and 
(ii) ${\bf D}^{(2)}_1 \leq_C {\bf D}^{(3)}_1$. 
Therefore, Lemma~\ref{lemma:L2} tells us
that $\bmu^2 \leq \bmu^3$. 
Indeed, $\bmu^3 = (0.9604 \ 0.9604 \ 0.9732 \ 
0.9732)$, which is larger 
than $\bmu^2$ from the previous example.
Accordingly, the PoCF 
decreases from 0.0414 to 0.0396, which 
represents approximately a 4.5 percent reduction 
in PoCF, 
despite very weak correlations in $p^{(3)}_1$;  
although $p^{(2)}_1$ and $p^{(3)}_1$ are close
(with Kullback-Leibler divergence 
$D_{{\rm KL}}(p^{(2)}_1 || p^{(3)}_1) =$ 0.0094), 
the likelihood of experiencing widespread
failures changes somewhat 
noticeably. This points
to possible {\em sensitivity} of PoEs
to the degree distributions, including
their dependence structure, 
in some cases.

\section{General Settings}
	\label{sec:General}

In Section~\ref{sec:Two}, we considered scenarios 
with two CSes and deterministic transmission 
of infections, and studied how the variability 
of degrees and dependence
between the two degrees of agents alter the
PoEs. In this section, we 
return to the general settings described
in Section~\ref{sec:Model} and examine 
how the degree distributions $(p_i; i \in 
\cS)$ shape the PoEs.

General settings pose additional challenges that 
we did not have to cope with in the simpler
two-CS scenarios. First, unlike in univariate
or bivariate cases, choosing a suitable stochastic 
order for comparing degree distributions becomes 
more problematic. The reason for this is that 
there are several different stochastic orders one 
can consider, which can be
viewed as extensions of a single
stochastic order for univariate rvs to
random vectors.
Second, perhaps more importantly, 
if the infection probabilities
$q_{i,j}$, $i, j \in \cS$, 
are not equal to one, even when two 
different sets of {\em degree} distributions 
$(p^{(\ell)}_i; \ i \in \cS)$, $\ell
= 1, 2$, can be ordered using some 
stochastic order, the associated {\em children}
distributions $(h^{(\ell)}_i; \ i \in 
\cS^+)$ are in general {\em not} guaranteed 
to preserve the ordering with respect to 
the same stochastic order. 

Consider two sets of degree distributions $\bp^{(\ell)} 
= (p^{(\ell)}_i; \ i \in \cS)$, $\ell = 1, 
2$. Let ${\bf D}^{(\ell)}_i = (D^{(\ell)}_{i,j}; \ 
j \in \cS)$, $\ell = 1, 2$, and $i \in \cS$,
be a random vector with distribution $p^{(\ell)}_i$. 
In order to make progress, we introduce the 
following assumption on internal failure 
probability functions $\wp^{in}_i$, 
$i \in \cS$. 

\begin{assm}	\label{assm:phi}
Assume that $\phi_i^\star: \N \to [0, \infty)$, 
where $\phi_i^\star(d) := d \cdot 
\wp^{in}_i(d, 1)$,  
is non-decreasing and satisfies
$\alpha \cdot	
\phi^\star_i(d_1) + (1 - \alpha)
\phi^\star_i(d_3)
\leq \phi^\star_i(d_2)$ for all $d_1
\leq d_2 \leq d_3$, where $\alpha
= (d_3 - d_2) / (d_3 - d_1)$.
\end{assm}

Roughly speaking, 
Assumption~\ref{assm:phi} states that an agent
has a higher {\em total} risk of experiencing a 
failure (due to the failure of one of the internal 
neighbors) with an increasing internal degree. 
An example of $\wp^{in}_i$ that satisfies Assumption
\ref{assm:phi} is $\wp^{in}_i(d, 1) = d^{-\beta}$ 
with $\beta \in [0, 1]$. Obviously, $\wp^{in}_i(d, 1) 
= \sum_{k} c_k \cdot d^{-\beta_k}$, where
$c_k > 0$ and $\beta_k \in [0, 1]$ for all $k$, 
also satisfies the assumption. We point out
that, as illustrated by the example, 
this assumption captures the heightened 
{\em aggregate} risk of failure seen by higher
degree agents, even though a single neighbor
might pose less risk (i.e., smaller 
$\wp^{in}_i(d, 1)$).

\subsection{The effects of dependence}

The dependence order we adopted in the
previous section, namely CO, 
can be generalized to random vectors 
consisting of more than two rvs: 
suppose that $\bX$ and $\bY$ are 
$n$-dimensional random vectors with 
$n > 2$. Then, $\bX$ is said to be
smaller than $\bY$ in CO
if $F_{\bX}(\bt) \leq F_{\bY}(\bt)$
and $\bar{F}_{\bX}(\bt) \leq 
\bar{F}_{\bY}(\bt)$ for all $\bt 
\in \R^n$, where $F_{\bX} = \bP{
\bX \leq \bt}$  and $\bar{F}_{\bX}
= \bP{\bX > \bt}$. One can verify
that these conditions imply 
that the marginal distributions 
are identical.

\subsubsection{Supermodular order}

In our study, we instead consider a dependence
order that is somewhat stronger than 
CO. This is called
{\em supermodular order} (SMO)
\cite{MullerStoyan}: a function 
$\xi: \R^n \to \R$ is called {\em 
supermodular}
if, for all $\bx, \by
\in \R^n$, 
\beqa
\xi(\bx) + \xi(\by) \leq
\xi(\bx \vee \by) + \xi(\bx \wedge \by), 
	\label{eq:supermodular}
\eeqa
where $\bx \wedge \by$ and $\bx \vee
\by$ are $(\min(x_i, y_i); \ i = 1, 2,
\ldots, n)$ and $(\max(x_i, y_i); \ i 
= 1, 2, \ldots, n)$, respectively. 
If the inequality in (\ref{eq:supermodular})
goes the other way, the function $\xi$
is called {\em submodular}. 

A random vector $\bX$ is smaller than 
a random vector $\bY$ in SMO 
($\bX \leq_{sm}
\bY$) if $\E{\xi(\bX)} \leq \E{\xi(\bY)}$
for all supermodular functions $\xi$
\cite{MullerStoyan}. SMO is a multivariate
positive dependence order that
satisfies the nine natural properties 
discussed in 
\cite[pp. 110-111]{MullerStoyan}. 
Furthermore, for bivariate cases, 
CO and SMO are {\em equivalent}. 
For $n > 2$, however, SMO implies CO, but 
they are no longer equivalent. Finally, 
if $\bX \leq_{sm} \bY$, we have
Cov($X_i, X_j$) $\leq$ Cov($Y_i, 
Y_j$) for all $i \neq j$.

The following theorem generalizes Lemma \ref{lemma:L2}. 
We will defer the proof of the theorem till after 
Theorem~\ref{thm:IDCV} in the subsequent subsection.

\begin{theorem}	\label{thm:SM}
Suppose that Assumption~\ref{assm:phi} holds
and ${\bf D}^{(1)}_i \leq_{sm} {\bf D}^{(2)}_i$ for
all $i \in \cS$. Let $\bmu^\ell$, $\ell = 1, 2$,
be the PoE vector under the degree distributions 
$(p^{(\ell)}_i; \ i \in \cS)$. Then, 
we have $\bmu^1 \leq \bmu^2$. 
\myskip
\end{theorem}


\subsection{The effects of variability}

In Section~\ref{sec:Two} with bivariate
degrees, we assumed {\em independence} of two 
degrees and studied the influence of 
variability of each degree on the PoEs
with the help of SSD (or ICV order). 
In this section, we adopt a 
stochastic order that allows us to 
examine the impact of variability
with a common dependence structure captured
by what is called {\em copula}~\cite{Copula}, 
without having to assume independence. 

\subsubsection{Copulas}

Suppose that $\bX$ is an $n$-dimensional
random vector with a joint distribution 
function $F_{\bX}$. A copula of $\bX$
(or associated with $F_{\bX}$) is a 
function $C_{\bX}: [0, 1]^n \to [0, 1]$ 
satisfying
\beqa
\myhb C_{\bX}(F_{X_1}(x_1), \ldots, F_{X_n}(x_n))
\myeq F_{\bX}(\bx) \ \mbox{ for all $\bx \in \R^n$},
	\label{eq:copula}
\eeqa
where $F_{X_i}$ is the 
marginal distribution of $X_i$, $i = 1, 2
\ldots, n$. For instance, a copula of
mutually independent $X_i$, $i = 1, 2
\ldots, n$, is a 
product function, i.e., for all 
$\bu \in [0, 1]^n$, we have  
$C_{\bX}(\bu) = \prod_{i=1}^n u_i$. 
Also, if $F_{X_i}$, $i = 1, 2, \ldots, n$,
are continuous, 
there is a unique copula that satisfies
(\ref{eq:copula}). 

It is clear from (\ref{eq:copula}) that 
a copula of a random vector captures most
of the dependence structure properties 
that do not depend on the marginal 
distributions.
In this sense, two random vectors with 
a common copula have similar dependence structure
among the comprising rvs. For a more detailed
discussion of copulas, we refer an 
interested reader to a manuscript
by Nelson~\cite{Copula}.

\subsubsection{Increasing directionally concave order}

A function $\xi: \R^n \to \R$ is said to 
{\em directionally concave} (DCV) if, for all 
$\bx_i \in \R^n$, $i = 1, 2, 3, 4$, with 
$\bx_1 \leq \bx_2, \bx_3 \leq \bx_4$
and $\bx_1 + \bx_4 = \bx_2 + \bx_3$, we have
$\xi(\bx_1) + \xi(\bx_4)
\leq \xi(\bx_2) + \xi(\bx_3)$. 
If the inequality goes the other way, the 
function is called {\em directionally 
convex} (DCX). Clearly, $\xi$ is DCV iff  
$-\xi$ is DCX. It turns out that 
a function $\xi$ is DCV (resp. DCX) iff 
it is {\em submodular} and {\em componentwise 
concave} (resp. {\em supermodular} and 
{\em componentwise convex})~\cite[p. 335]{SO}. 

Let $\bX$ and $\bY$ be two $n$-dimensional random 
vectors. Random vector $\bX$ precedes $\bY$ in 
DCV order ($\bX \leq_{dcv} \bY$) if $\E{\xi(\bX)} 
\leq \E{\xi(\bY)}$ for all DCV functions $\xi$. 
Note that $\bX \leq_{dcv} \bY$ iff
$\bY \leq_{dcx} \bX$. If the inequality
is required only for {\em increasing} 
DCV functions, we say that $\bX$ precedes $\bY$ 
in increasing DCV (IDCV) order. As expected,
$\bX \leq_{idcv} \bY$ iff $\bY$ precedes $\bX$ in 
{\em decreasing} DCX (DDCX) order
($\bY \leq_{ddcx} \bX$). 

If $\bX \leq_{idcv} \bY$, 
then $X_i \leq_{icv} Y_i$ for all $i = 1, 
2, \ldots, n$. Therefore, $Y_i$, $i = 1, 2, 
\ldots, n$, are in a way more 
{\em predictable} than $X_i$, $i = 1, 2, 
\ldots, n$. 
As pointed out in \cite[p. 135]{MullerStoyan}, 
the DCV (or DCX) order goes one
step further and allows us to compare 
random vectors with a {\em common copula}, but 
{\em with different variability in the 
marginals}. Moreover, 
an example is provided to illustrate that
convex order is not suitable for this purpose.

Utilizing the IDCV order (or, equivalently, DDCX
order), the following theorem sheds some light 
on how the variability of agents' degrees influences
the PoEs, even when the mean degrees stay fixed.

\begin{theorem}	\label{thm:IDCV}
Suppose that Assumption~\ref{assm:phi} holds
and ${\bf D}^{(1)}_i \leq_{idcv}
{\bf D}^{(2)}_i$ with $\E{{\bf D}^{(1)}_i}
= \E{{\bf D}^{(2)}_i}$ for
all $i \in \cS$. Let $\bmu^\ell$, $\ell = 1, 2$,
be the PoE vector under the distributions 
$(p^{(\ell)}_i; \ i \in \cS)$. Then, 
we have $\bmu^2 \leq \bmu^1$. 
\end{theorem}
\begin{proof}
A proof of Theorem~\ref{thm:IDCV} is given in 
Section~\ref{appen:IDCV}. 
\end{proof}
 
As explained before, when ${\bf D}^{(1)}_i 
\leq_{idcv} {\bf D}^{(2)}_i$ for 
some $i \in \cS$, the degrees of CS $i$ agents
with the distribution $p^{(1)}_i$ have greater 
variability than with the distribution $p^{(2)}_i$. 
As a result, Theorem~\ref{thm:IDCV}
can be viewed as a generalization of Lemma
\ref{lemma:L1}.

We are now ready to provide the proof of
Theorem~\ref{thm:SM}. 
\myskip 

\hspace{0.07in} 
{\em Proof of Theorem~\ref{thm:SM}:}
Recall that a function is DCX iff it is both 
supermodular and componentwise convex. Thus, 
it is obvious that SMO implies DCX order, hence
DDCX order. For this reason, if ${\bf D}^{(2)}_i 
\leq_{sm} {\bf D}^{(1)}_i$, then ${\bf D}^{(2)}_i 
\leq_{ddcx} {\bf D}^{(1)}_i$ or, equivalently, 
${\bf D}^{(1)}_i \leq_{idcv} {\bf D}^{(2)}_i$.
Now Theorem
\ref{thm:SM} follows from Theorem~\ref{thm:IDCV}.

\subsection{Comparison on the basis of 
	children distributions}

In Theorems~\ref{thm:SM} and \ref{thm:IDCV}, 
the inequalities in stochastic orders
are imposed on the degree distributions 
$(p^{(\ell)}_i; \ i \in \cS)$. In some cases, 
however, it may be possible to estimate the 
children distribution. If we could directly
compare the {\em children} distributions
(as we implicitly 
did in Section~\ref{sec:Two} for 
two symmetric CSes), we can prove a stronger 
result than Theorem~\ref{thm:IDCV}.

Suppose that $\bX$ and $\bY$ are two 
$n$-dimensional random vectors. We say 
that $\bX$ is smaller than $\bY$ in 
{\em Laplace transform} (LT) order
($\bX \leq_{LT} \bY$) if 
$\E{\exp\left( - \bs \bX^T \right)}
\geq \E{\exp\left( - \bs \bY^T \right)}
\ \mbox{ for all } \bs > {\bf 0}$.

One can easily show that, for every 
$\bs > {\bf 0}$, the function $\xi_{\bs}:
\R^n \to \R$, where
$\xi_{\bs}(\bx) = \exp(- \bs \bx^T)$ 
for $\bx \in \R^n$, is DDCX: since $\xi_{\bs}$
is twice differentiable, it is DCX iff
$\frac{\partial^2}{\partial x_i \partial x_j}
\xi_{\bs}(\bx) \geq 0$ for all $\bx \in 
\R^n$~\cite[Theorem 3.12.2, p. 132]{MullerStoyan}.
For all $i, j \in \{1, 2, \ldots, n\}$, we have 
\beqan
\frac{\partial^2}{\partial x_i \partial x_j}
	\xi_{\bs}(\bx)
\myeq s_j s_i \exp\left( - \bs \bx^T
	\right) > 0 \ \mbox{ for all } 
	\bx \in \R^n.
\eeqan
Clearly, $\xi_{\bs}(\bx)$ is decreasing in 
$\bx$ and, hence, is DDCX. This tells us
that if $\bX \leq_{ddcx} \bY$, then 
$\bX \leq_{LT} \bY$.  

Suppose that $(h^{(\ell)}_i; \ i \in \cS^+)$, 
$\ell = 1, 2$, 
are the children distributions under degree
distributions $(p^{(\ell)}_i; 
\ i \in \cS)$. Let $\bO^{(\ell)}_i
= (O^{(\ell)}_{i,j}; \ j \in \cS^+)$
be the random children vector with distribution
$h^{(\ell)}_i$, $\ell = 1, 2$, and $i \in \cS^+$. 
The following theorem holds under a much weaker 
condition (namely, LT order) than IDCV order
without Assumption~\ref{assm:phi}.

\begin{theorem}	\label{thm:LT}
Suppose that ${\bf O}^{(1)}_i \leq_{LT}
{\bf O}^{(2)}_i$ for all $i \in \cS^+$. 
Let $\bmu^\ell$, $\ell = 1, 2$,
be the PoE vector under the children 
distributions 
$(h^{(\ell)}_i; \ i \in \cS^+)$. Then, 
$\bmu^2 \leq \bmu^1$. 
\myskip
\end{theorem}
\begin{proof}
Please see Section~\ref{appen:LT} for a proof. 
\end{proof}

\change{
\subsection{Discussion}	
	\label{subsec:Discussion}
	
In this subsection, we briefly discuss some of
modeling assumptions and their roles. 

$\bullet$ {\bf Asymmetric inter-CS dependence and
symmetric intra-CS dependence --}
In our model outlined in Section~\ref{sec:Model}, 
we purposely assumed somewhat different 
intra-CS dependence and inter-CS dependence
among agents. More specifically, we assumed
that intra-CS dependence is symmetric, while
inter-CS dependence is asymmetric. 

Suppose that the inter-CS dependence is symmetric
instead. Then, the inter-CS infection probabilities 
$q_{i,j}$, $i \neq j$, can be computed in a manner 
similar to that of $q_{i,i}$, $i \in \cS$, 
as outlined in Section~\ref{subsec:InfProb}. Since 
Theorems~\ref{thm:SM} - \ref{thm:LT} hold
with symmetric internal dependence, it is 
not surprising that 
the same results hold with symmetric
inter-CS dependence if $\wp^{ex}_{j,i}(d, 1)$, 
$j \neq i$, satisfy Assumption~\ref{assm:phi}.
Similarly, if the intra-CS dependence is 
asymmetric, under the same assumption that
the conditional degree distributions of failed
agents are similar to the prior degree
distributions $p_i$, $i \in \cS$, 
the results still hold.
These observations suggest that our findings 
are true under more general settings and
are not sensitive to specific modeling 
assumptions introduced in the study.

$\bullet$ {\bf Effects of failure probability
functions --} 
Recall that the functions $\wp^{in}_i$ and 
$\wp^{ex}_{j,i}$ are used to model 
the vulnerability of agents to the failures of 
those that support them. 
It is clear from Section~\ref{subsec:InfProb}
that, except for Assumption
\ref{assm:phi} on $\wp^{in}_i$, 
our main findings in Theorems
\ref{thm:SM} - \ref{thm:LT} do not impose 
any other conditions on $\wp^{in}_i$ and 
$\wp^{ex}_{j,i}$. In particular, although 
it is reasonable to
expect agents with larger incoming
degrees to be less susceptible to the failure
of a {\em single} supporting agent, the 
functions $\wp^{ex}_{j,i}$ need not be 
monotonic. 
This suggests that our reported findings are
{\em insensitive} to the exact 
choices of these functions as long as 
Assumption~\ref{assm:phi} is met by
$\wp^{in}_i$, $i \in \cS$. 

But, we also note that this may no longer be 
true if the dependence graph exhibits 
assortativity, i.e., 
the degrees of neighbors are correlated. 
We do not study the impact of 
assortativity here, and refer an
interested reader to \cite{La_TON17}
for a study of its influence in 
a single homogeneous network with strategic
agents.

}

\section{Proofs of Main results}
	\label{sec:Proofs}

This section provides the proofs of our main
results. A reader who is not interested in the
proofs can safely skip the section. 

\subsection{Proof of Lemma~\ref{lemma:L1}}
	\label{appen:L1}

Let $\bff^{\ell}$, $\ell = 1, 2$, be the generating 
functions corresponding to the degree distributions 
$p^{(\ell)}_1$, $\ell = 1, 2$. 
Using the definition of the generating function, 
we obtain
\beqa
f_1(\bs)
\myeq \sum_{\bd \in \Z_+^2} p_1(\bd) \ s_3^{d_1}
	s_2^{d_2} \ \mbox{ and } 
	\label{eq:L1-1} \\
f_3(\bs)
\myeq \sum_{\bd \in \Z_+^2} p_1(\bd+{\bf e}_1) \ s_3^{d_1}
	s_2^{d_2}.
	\label{eq:L1-3} 
\eeqa
The expressions for $f_2(\bs)$ and $f_4(\bs)$ can be 
easily obtained using the assumed symmetry of degree 
distributions $p_1$ and $p_2$. 
Recall that, from the assumed symmetry of the degree 
distributions, we know $\mu_1 = 
\mu_2$ and $\mu_3 = \mu_4$.

Substituting $\mu_1$ for $\mu_2$ in (\ref{eq:L1-1}) for 
the PoEs and using the assumed independence of the two
degrees of an agent, 
\beqan
\myhb \mu_1 
\myeq \sum_{\bd \in \Z_+^2} p_1(\bd) 
	\ \mu_3^{d_1} \mu_1^{d_2} 
		\label{eq:L1-5} \\
\myeq \left( \sum_{d_1 \in \Z_+} p_{1,1}(d_1) 
		\ \mu_3^{d_1} \right)
	\left( \sum_{d_2 \in \Z_+} p_{1,2}(d_2) 
		\ \mu_1^{d_2} \right). 
	\nonumber
\eeqan

Clearly, for any fixed $s \in (0, 1)$, $s^x$
is a decreasing, convex function of $x$ or, 
equivalently, $-s^x$ is an increasing, concave 
function of $x$. Because $D^{(1)}_{1,j}
\leq_{SSD} D^{(2)}_{1, j}$, $j = 1, 2,$
we have, for all $\bmu < {\bf 1}$,  
\beqan
&& \myhb \left( \sum_{d_1 \in \Z_+} p^{(2)}_{1,1}(d_1) 
		\ \mu_3^{d_1} \right)
	\left( \sum_{d_2 \in \Z_+} p^{(2)}_{1,2}(d_2) 
		\ \mu_1^{d_2} \right) \lb
\myleq \left( \sum_{d_1 \in \Z_+} p^{(1)}_{1,1}(d_1) 
		\ \mu_3^{d_1} \right)
	\left( \sum_{d_2 \in \Z_+} p^{(1)}_{1,2}(d_2) 
		\ \mu_1^{d_2} \right). 
\eeqan
Repeating the same steps, starting with 
(\ref{eq:L1-3}), yields
\beqan
&& \myhb \left( \sum_{d_1 \in \Z_+} p^{(2)}_{1,1}(d_1+1) 
		\ \mu_3^{d_1} \right)
	\left( \sum_{d_2 \in \Z_+} p^{(2)}_{1,2}(d_2) 
		\ \mu_1^{d_2} \right) \lb
\myleq \left( \sum_{d_1 \in \Z_+} p^{(1)}_{1,1}(d_1+1) 
		\ \mu_3^{d_1} \right)
	\left( \sum_{d_2 \in \Z_+} p^{(1)}_{1,2}(d_2) 
		\ \mu_1^{d_2} \right). 
\eeqan

Since these two inequalities hold for all $\bmu < {\bf 1}$, 
we get 
\beqa
\bff^2(\bmu^1) 
\myleq \bff^1(\bmu^1) = \bmu^1. 
	\label{eq:L1-7}
\eeqa
Corollary 2~\cite[p. 42]{Harris} tells us that
the inequality $\bff^2(\bmu^1) \leq \bmu^1$ 
in (\ref{eq:L1-7}) means $\bmu^2 \leq \bmu^1$. 
This completes the proof.

%
%

\subsection{Proof of Theorem~\ref{thm:IDCV}}
	\label{appen:IDCV}
	
In order to prove the theorem, we will make
use of the following lemma. Its proof is
provided in Section~\ref{appen:lemma-hidcv}.
\myskip

\begin{lemma}	\label{lemma:hidcv}
Under Assumption~\ref{assm:phi}, we have
${\bf O}^{(1)}_i \leq_{idcv} {\bf O}^{(2)}_i$
for all $i \in \cS^+$, 
where ${\bf O}^{(\ell)}_i$, $\ell = 1, 2$,
is a random vector with the children distribution
$h^{(\ell)}_i$.
\myskip
\end{lemma}
 
We know from Corollary 1~\cite[p. 42]{Harris}
that the only solutions of (\ref{eq:fixed}) 
in the unit cube are $\boldsymbol{\mu}$ and 
${\bf 1}$. However, when $\rho(\bM) > 1$, 
we have $\boldsymbol{\mu} < {\bf 1}$. 

Let $\bff^\ell$ be the generating function 
under the children distributions $\bh^{(\ell)}
= (h^{(\ell)}_i; \ i \in \cS^+)$, 
$\ell = 1, 2$. Then, $\bff^1(\bmu^1) = \bmu^1 
< {\bf 1}$. From the definition of the generating 
function, 
\beqan
f^1_i(\bmu^1)
\myeq \sum_{\bo \in \Z^{2N_S}_+} h^{(1)}_i(\bo)
	\prod_{j \in \cS^+} \left( \mu^1_j \right)^{o_j}
		\ \mbox{ for all } i \in \cS^+.
\eeqan

Since ${\bf O}^{(1)}_i \leq_{idcv}{\bf O}^{(2)}_i$ or, 
equivalently, ${\bf O}^{(2)}_i \leq_{ddcx} 
{\bf O}^{(1)}_i$ for all $i \in \cS^+$, 
we can show $\bff^2(\bs) \leq \bff^1(\bs)$ for all $\bs
< {\bf 1}$ as follows: for any ${\bf 0} < \bs < {\bf 1}$, 
define $\psi_{\bs}: \R^{2 N_S} \to \R$, where 
$\psi_{\bs}(\bx)
= \exp\left( - \sum_{i \in \cS^+}\log(s_i^{-1}) 
	\ x_i \right).$
According to Theorem 3.12.2~\cite[p. 132]{MullerStoyan}, 
$\psi_{\bs}$ is DCX iff $\partial^2
\psi_{\bs}(\bx) / \partial x_i \partial x_j \geq 0$ 
for all $\bx \in \R^{2 N_S}$ and $i, j \in \cS^+$. 
Clearly, 
\beqan
\frac{\partial^2}{\partial x_i \partial x_j} 
	\psi_{\bs}(\bx) = \log(s_i^{-1}) \ 
		\log(s_j^{-1}) \ \psi_{\bs}(\bx) 
		\geq 0
\eeqan
for all $\bx \in \R^{2 N_S}$ and $i, j \in \cS^+$, 
and $\psi_{\bs}$ is DCX. 

Together with $\bO^{(2)}_i \leq_{ddcx} \bO^{(1)}_i$
and $\bmu^1 < {\bf 1}$, this tells us that, for all
$i \in \cS^+$, 
\beqan
\myhb f^2_i(\bmu^1)
\myeq \mathbbm{E}\left[ \psi_{\bmu^1}({\bf O}^{(2)}_i) \right]
\leq \mathbbm{E}\left[ \psi_{\bmu^1}({\bf O}^{(1)}_i) \right]
= \mu^1_i,
	\label{eq:appen1} 
\eeqan 
and $\bff^2(\bmu^1) \leq \bmu^1$. Corollary 2 
\cite[p. 42]{Harris} states that, if 
${\bf 0} \leq \bff^2(\bmu^1) \leq \bmu^1 \leq {\bf 1}$, 
then $\bmu^2 \leq \bmu^1$.  This completes the proof
of the theorem.

\subsection{Proof of Lemma~\ref{lemma:hidcv}}
	\label{appen:lemma-hidcv}

We will first prove the lemma for $i \in \cS$
and then for $i \in \cS^+ \setminus \cS
= \{N_S + 1, \ldots, 2 N_S\}$. 

$\bullet$ ${\bf O}^{(1)}_i \leq _{idcv} 
{\bf O}^{(2)}_i$, $i \in \cS$:
We prove the claim using the definition of 
IDCV order. 
Suppose that $g: \R^{2 N_S} \to \R$ is IDCV. 
Recall that this is equivalent to saying that
$-g$ is DDCX. For notational convenience, 
for each $j \in \cS$, define 
\beqan 
\cO_j 
:= \big\{ \bo \in \Z^{2N_S}_+ \ | \ o_{j'} = 0 \mbox{
	for all } j' \notin \cS^+_j \big\}.
\eeqan

For $\ell = 1, 2$, 
\beqa
\mathbbm{E}\left[ g(\bO^{(\ell)}_i) \right]
\myeq \sum_{\bo \in \cO_i} g(\bo) \ h^{(\ell)}_i(\bo)
		\label{eq:appen2-1} \\
\myeq \sum_{\bo \in \cO_i} g(\bo) \left[
	\sum_{\bd \geq \bar{\bd}(\bo)} p^{(\ell)}_i(\bd) 
		\right. \lb 
&& \times \left. \prod_{j \in \cS} 
			{{d_j}\choose{\bar{d}_j(\bo)}}
			\left(q^\ell_{i,j}\right)^{\bar{d}_j(\bo)}
			\left( \bar{q}^\ell_{i,j} 
				\right)^{d_j - \bar{d}_j(\bo)}
			\right],  
	\nonumber
\eeqa
where the mapping $\bar{\bd}$ was defined in (\ref{eq:bard}), 
$q^\ell_{i,j} = q_{i,j}$ for $j \neq i$, and 
$q^\ell_{i,i}$ is the probability that an internal 
neighbor of a CS $i$ agent will fall victim to 
infection, under the internal degree distribution 
$p^{(\ell)}_{i,i}$, which is given in (\ref{eq:qii-1}).

We shall prove $\mathbbm{E}\left[ g(\bO^{(1)}_i) 
\right] \leq \mathbbm{E}\left[ g(\bO^{(2)}_i) \right]$
in two steps. First, we will show 
\beqa
\mathbbm{E}\left[ g(\bO^{(1)}_i) \right] 
\myeq \sum_{\bo \in \cO_i} g(\bo) \left[
	\sum_{\bd \geq \bar{\bd}(\bo)} p^{(1)}_i(\bd) 
		\right. \lb 
&& \myb \times \left. \prod_{j \in \cS} 
			{{d_j}\choose{\bar{d}_j(\bo)}}
			\left(q^1_{i,j}\right)^{\bar{d}_j(\bo)}
			\left( \bar{q}^1_{i,j} 
				\right)^{d_j - \bar{d}_j(\bo)}
			\right]
	\lb
\myleq \sum_{\bo \in \cO_i} g(\bo) \left[
	\sum_{\bd \geq \bar{\bd}(\bo)} p^{(2)}_i(\bd) 
		\right. \label{eq:appen2-0} \\
&& \myb \times \left. \prod_{j \in \cS} 
			{{d_j}\choose{\bar{d}_j(\bo)}}
			\left(q^1_{i,j}\right)^{\bar{d}_j(\bo)}
			\left( \bar{q}^1_{i,j} 
				\right)^{d_j - \bar{d}_j(\bo)}
			\right].
	\nonumber
\eeqa
Then, we will demonstrate that $(\ref{eq:appen2-0})
\leq \mathbbm{E}\left[ g(\bO^{(2)}_i) \right]$.
If this is true for all IDCV functions $g$, it implies
${\bf O}^{(1)}_i \leq_{idcv} {\bf O}^{(2)}_i$ by 
definition.

After interchanging the order of the two summations
in (\ref{eq:appen2-1}), we get
\beqan
&& \myhb \mathbbm{E}\left[ g(\bO^{(\ell)}_i) \right]
	= \sum_{\bd \in \Z^{N_S}_+} p^{(\ell)}_i(\bd) \lb 
&& \myhb \times \Big[
	\underbrace{\sum_{\bo \in \cO_i: \bar{\bd}(\bo) \leq \bd} g(\bo) 
		\prod_{j \in \cS} {{d_j}\choose{\bar{d}_j(\bo)}}
			\left(q^\ell_{i,j}\right)^{\bar{d}_j(\bo)}
			\left( \bar{q}^\ell_{i,j} \right)^{d_j 
				- \bar{d}_j(\bo)}}_\text{
				\large{$=: g^*_\ell(\bd)$}} \Big]. 
\eeqan
Note that $g^*_\ell(\bd)$ is $\mathbbm{E}\left[ g(\bY(\bd, 
\ell)) \right]$, where $\bY(\bd, \ell)$ is a vector 
consisting of 
$N_S$ mutually independent binomial rvs. In other words, 
with a little abuse of notation, the $j$-th element
$Y_j(d_j, \ell)$ is a Binomial($d_j, q^\ell_{i,j}$) rv 
and, hence, can be viewed as a sum of $d_j$ i.i.d. 
Bernoulli($q^\ell_{i,j}$) rvs. 

In order to finish the proof of the first step, we make
use of the following lemma. 
\myskip

{\em Lemma 2.17 \cite{MeesterShanti}:} Let $\left( S^{(i)}_j, 
j = 1, 2, ... \right)$ be independent sequences of i.i.d. 
nonnegative rvs. Suppose $\theta: \R^m \to \R$
is increasing DCX (resp. IDCV). Then, $\zeta: \N^m \to \R$ 
defined by $\zeta(\bn) = \E{\theta \left( \sum_{j=1}^{n_1} 
S^{(1)}_j, \ldots, \sum_{j=1}^{n_m} S^{(m)}_j \right)}$ 
is increasing DCX (resp. IDCV). 
\myskip

Since the function $g$ was assumed IDCV and $\bY(\bd, 1)$ 
consists of mutually independent binomial rvs (each of which
is a sum of i.i.d. Bernoulli rvs), the above 
lemma tells us that $g^*_1$ is IDCV. Since $g^*_1$ is IDCV 
and, with a little abuse of notation, 
$p^{(1)}_i \leq_{idcv} p^{(2)}_i$, we have 
\beqan
&& \myhb \mathbbm{E}\left[ g(\bO^{(1)}_i) \right]
= \mathbbm{E}\left[ g^*_1({\bf D}^{(1)}_i) \right] 
= \sum_{\bd \in \Z^{N_S}_+} p^{(1)}_i(\bd) \ g^*_1(\bd) \lb
\myleq \sum_{\bd \in \Z^{N_S}_+} p^{(2)}_i(\bd) \ g^*_1(\bd)
= \E{g^*_1({\bf D}^{(2)}_i)}
= (\ref{eq:appen2-0}).
	\label{eq:appen2-2}
\eeqan
This proves the first step. 
	
In order to prove the second step
\beqan 
(\ref{eq:appen2-0}) 
\myleq \mathbbm{E}\left[ g(\bO^{(2)}_i) \right] 
= \E{g^*_2({\bf D}^{(2)}_i)}
= \sum_{\bd \in \Z^{N_S}_+} p^{(2)}_i(\bd) \ g^*_2(\bd), 
\eeqan
it suffices
to show that $q^1_{i,j} \leq q^2_{i,j}$ for all $j \in 
\cS$; if this is true, for all $\bd \in \Z_+^{N_S}$, 
$Y_j(d_j,1) \leq_{st} Y_j(d_j, 2)$ for all $j \in \cS$, 
where $\leq_{st}$ denotes the inequality with respect
to {\em usual stochastic order}.  
Since $Y_j(d_j, \ell)$, $j \in \cS$, are mutually
independent, Theorem 3.3.8~\cite[p. 93]{MullerStoyan}
tells us ${\bf Y}(\bd, 1) \leq_{st} {\bf Y}(\bd, 
2)$. Because the function $g$ is assumed increasing, 
this implies 
\beqan
g^*_1(\bd) 
= \E{g(\bY(\bd, 1))} 
\leq \E{g(\bY(\bd, 2))}
= g^*_2(\bd) 
\eeqan
from the definition of usual stochastic order
\cite{MullerStoyan}. 

First, recall that $q^1_{i,j} = q^2_{i,j} = q_{i,j}$ 
for all $j \neq i$ because they are assumed fixed. 
Thus, we only need to show $q^1_{i,i} \leq q^2_{i,i}$. 
Using the definition in (\ref{eq:qii-1}), 
\beqan
q^1_{i,i}
\myeq \frac{1}{d^{(1)}_{i,\avg}} 
	\sum_{d \in \N} p^{(1)}_{i,i}(d)
	\left( d \cdot \wp^{in}_i(d, 1) \right) \lb
\myeq \frac{1}{d^{(1)}_{i,\avg}} 
	\sum_{d \in \N} p^{(1)}_{i,i}(d)
	\ \phi_i^*(d). 
\eeqan
Because $\phi_i^*$ satisfies Assumption~\ref{assm:phi} 
and $p^{(1)}_i \leq_{idcv} p^{(2)}_i$, which implies 
$p^{(1)}_{i,i} \leq_{icv} p^{(2)}_{i,i}$, we obtain
\beqan
q^1_{i,i} 
\myeq \frac{1}{d^{(1)}_{i,\avg}} 
	\sum_{d \in \N} p^{(1)}_{i,i}(d) \ \phi_i^*(d) \lb
\myleq \frac{1}{d^{(2)}_{i,\avg}} 
	\sum_{d \in \N} p^{(2)}_{i,i}(d) \ \phi_i^*(d) 
	= q^2_{i,i}. 
\eeqan
Recall that $d^{(1)}_{i,\avg} = d^{(2)}_{i,\avg}$
from the assumption in the theorem. 
This proves $\mathbbm{E}\left[ g(\bO^{(1)}_i) 
\right] \leq \mathbbm{E}\left[ g(\bO^{(2)}_i) \right]$.
Since this inequality holds for every IDCV function $g$
(with well defined expectations), we have 
${\bf O}^{(1)}_i \leq_{idcv} {\bf O}^{(2)}_i$. 
\myskip

$\bullet$ ${\bf O}^{(1)}_i \leq_{idcv} 
{\bf O}^{(2)}_i$, $i \in \cS^+ \setminus
\cS$: First, note that if $\varphi:\R^{2 N_S} \to \R$ is IDCV, 
then so is $\varphi^*_{\be}: \R^{2 N_S} \to \R$ 
with $\varphi^*_{\be}(\bx)
= \varphi(\bx + \be)$ for all $\be \in \Z^{2 N_S}$. This
follows directly from the observation that 
$\varphi^*_{\be}$ satisfies the characterization
(ii) and (iii) of DCV functions 
in Theorem 3.12.2 \cite[p.132]{MullerStoyan}.

This tells us that $\tilde{p}^{(\ell)}_i$, $\ell = 1, 2$, 
with $\tilde{p}^{(\ell)}_i(\bd) := 
p^{(\ell)}_i(\bd +\be_i)$ satisfy $\tilde{p}^{(1)}_i
\leq_{idcv} \tilde{p}^{(2)}_i$. The claim
that ${\bf O}^{(1)}_i \leq_{idcv} {\bf O}^{(2)}_i$ 
now follows from the proof of the previous case.

\subsection{Proof of Theorem~\ref{thm:LT}}
	\label{appen:LT}

From the definition of $\bmu^2$, we have
$\bff^2(\bmu^2) = \bmu^2 < {\bf 1}$.
Using the given children distributions, 
for all $i \in \cS$,
\beqa
&& \myhb f^2_i(\bmu^2) = \mu^2_i
= \sum_{\bo \in \cO_i} 
	h^{(2)}_i(\bo) 
	\prod_{j \in \cS^+_i} \left( 
		\mu^2_{j} \right)^{o_j} \lb 
\myeq \sum_{\bo \in \cO_i} 
	h^{(2)}_i(\bo) \ \exp\left( -
		\sum_{j \in \cS^+_i} o_j
		\pi^2_j \right),
	\label{eq:LT-1}
\eeqa
where $\pi^2_j = \log(1 / \mu^2_j) > 0$, 
$j \in \cS^+$. 
Similarly, for all $i^+ \in \cS^+ \setminus
\cS$, 
\beqa
&& \myhb f^2_{i^+}(\bmu^2) = \mu^2_{i^+}
= \sum_{\bo \in \cO_i} 
	h^{(2)}_{i^+}(\bo) 
	\prod_{j \in \cS^+_i} \left( 
		\mu^2_{j} \right)^{o_j} \lb 
\myeq \sum_{\bo \in \cO_i} 
	h^{(2)}_{i^+}(\bo) \ \exp\left( - 
	\sum_{j \in \cS^+_i} o_j
		\pi^2_j \right).
	\label{eq:LT-2}
\eeqa

Recall that ${\bf O}^{(1)}_i \leq_{LT} {\bf O}^{(2)}_i$ 
for all $i \in \cS^+$. Thus, because $\pi^2_j$, 
$j \in \cS^+_i$, 
are all positive, we get
\beqan
(\ref{eq:LT-1})
\myleq \sum_{\bo \in \cO_i} 
	h^{(1)}_i(\bo) \ \exp\left( - 
	\sum_{j \in \cS^+_i} o_j
		\pi^2_j \right) 
= f^1_i(\bmu^2),  
\eeqan
which yields $\mu^2_i \leq f^1_i(\bmu^2)$ for all 
$i \in \cS$. Following the same argument 
starting with (\ref{eq:LT-2}) gives
$\mu^2_{i^+} \leq f^1_{i^+}(\bmu^2)$ for all 
$i^+ \in \cS^+ \setminus \cS$. Together, we obtain
$\bmu^2 \leq \bff^1(\bmu^2)$. Corollary
2 \cite[p. 42]{Harris} now tells us $\bmu^2 
\leq \bmu^1$, completing the proof of the 
theorem.

\section{Conclusion}
	\label{sec:Conclusion}

We investigated the impact of variability and 
correlations in degrees of agents on the robustness
of interdependent systems. Our findings suggest
that they both can have significant influence
on the likelihoods of having catastrophic 
failures in complex systems comprising multiple
heterogeneous systems via dependency among
the agents. In particular, our results revealed
that both increasing variability and positive
dependence render the system more robust
against random failures. 

We are currently working to incorporate other
graph properties displayed by both 
natural and engineered systems, such as
assortativity and clustering, and to understand
their role in the robustness of interdependent
systems. Our goal is to identify a suitable
way of imposing a partial order on the 
underlying dependence graphs and compare their
resilience against both random and targeted
attacks.

\end{document}